\documentclass[a4paper,10pt]{article}
\usepackage{a4wide}
\usepackage[utf8]{inputenc}
\usepackage[T1]{fontenc}

\usepackage{subcaption}
\usepackage{amsmath,amssymb}
\usepackage{amsthm}
\usepackage{enumerate} 
\usepackage{authblk}

\usepackage[round]{natbib}
\bibpunct[, ]{(}{)}{;}{a}{}{,}

\usepackage[svgnames]{xcolor}
\usepackage[pdftex]{hyperref}
\hypersetup{
    colorlinks=true,
    colorlinks,
    linkcolor=Navy,
    citecolor=Navy,
    urlcolor=Navy
}
\definecolor{highlightNEW}{named}{black}

\usepackage{graphicx}
\graphicspath{{fig/}}
\usepackage{epstopdf}
\usepackage{booktabs}
\usepackage{multirow}
\usepackage{float}

\newtheorem{theorem}{Theorem}[section] 
\newtheorem{example}[theorem]{Example} 
\newtheorem{hypothesis}[theorem]{Hypothesis} 
\newtheorem{lemma}[theorem]{Lemma} 
\newtheorem{proposition}[theorem]{Proposition} 
\newtheorem{remark}[theorem]{Remark} 

\newcommand{\doi}[1]{DOI~\href{\detokenize{http://dx.doi.org/#1}}{\detokenize{#1}}}

\renewcommand{\d}{\,\mathrm{d}}

\newdimen\CdotAxis
\newcommand*{\CdotAux}[3]{%
  {%
    \settoheight\CdotAxis{$#2\vcenter{}$}%
    \sbox0{%
      \raisebox\CdotAxis{%
        \scalebox{#1}{%
          \raisebox{-\CdotAxis}{%
            $\mathsurround=0pt #2#3$%
          }%
        }%
      }%
    }%
    \dp0=0pt %
    \sbox2{$#2\bullet$}%
    \ifdim\ht2<\ht0 %
      \ht0=\ht2 %
    \fi
    \sbox2{$\mathsurround=0pt #2#3$}%
    \hbox to \wd2{\hss\usebox{0}\hss}%
  }%
}
\makeatletter
\def\mathcolor#1#{\@mathcolor{#1}}
\def\@mathcolor#1#2#3{%
  \protect\leavevmode
  \begingroup
    \color#1{#2}#3%
  \endgroup
}
\makeatother

\let\oldalpha\alpha
\renewcommand{\alpha}{\mathcolor{highlightNEW}{\oldalpha}}

\newcommand{\email}[1]{\href{mailto:#1}{#1}}

\newcommand\blfootnote[1]{%
  \begingroup
  \renewcommand\thefootnote{}\footnote{#1}%
  \addtocounter{footnote}{-1}%
  \endgroup
}

\title{\textcolor{Navy}{\textsc{Convexity adjustments \`a la Malliavin}}}

\author[1,2]{David Garcia-Lorite\thanks{Corresponding author, \email{david.garcia.lorite@gmail.com}}}
\author[3]{Ra\'{u}l Merino}

\affil[1]{CaixaBank, Quantitative Analyst Team, Plaza de Castilla, 3, 28046 Madrid, Spain,}
\affil[2]{Facultat de Matem\`{a}tiques i Inform\`{a}tica, Universitat de Barcelona, \authorcr Gran Via 585, 08007 Barcelona, Spain,\vspace*{3pt}}
\affil[3]{VidaCaixa S.A., Market Risk Management Unit, \authorcr C/Juan Gris, 2-8, 08014 Barcelona, Spain.}

\date{\normalfont\small\today}

\begin{document}

\maketitle
\begin{abstract}
In this paper, we develop a novel method based on Malliavin calculus to find an approximation for the convexity adjustment for various classical interest rate products. Malliavin calculus provides a simple way to get a template for the convexity adjustment. We find the approximation for Futures, OIS Futures, FRAs, and CMSs under a general family of the one-factor Cheyette model. We have also seen the excellent quality of the numerical accuracy of the formulas obtained.
\end{abstract}
%\keywords{V--}
%\ccode{MSC classification}{--}
%\ccode{JEL classification}{--}

%\TODO{Remove ToC in final version}
%\tableofcontents
%\clearpage
\blfootnote{The authors would like to thank Elisa Al\`os, Dariusz Gatarek, and Josep Vives for helpful discussions.}

\section{Introduction}
Mathematical finance aims to find a methodology to price consistently all the instruments quoted in the market. When working with fixed income derivatives, a classic research topic is the introduction of a price adjustment to achieve this. This adjustment is called convexity adjustment. It is non-linear and depends on the interest rate model.  

There are several reasons to include this type of adjustment. One of them is to incorporate futures on the yield curve construction. Futures and other fixed-income instruments are quoted differently. The firsts are linear against the yield, but the others are not. Therefore, the changes in value and yield of different contracts are different. This difference will depend on the volatility and correlation of the yield curve.

But it is not the only one. The fixed-income market has several features changing the schedule of payments. For example, in a swap in arrears, the floating coupon fixing and payment are on the same date. Or in a CMS swap, the floating rate is linked to a rate longer than the floating length. Any customization of an interest rate product based on changing time, currency, margin, or collateral will require a convexity adjustment. Deep down, by making these changes, we are mixing the martingale measures. 

Convexity adjustments have become popular again. Not only by the increase in volatility in the markets. In addition, as a consequence of the transition in risk-free rates from the IBOR (InterBank Offered Rates) indices to the ARR (Alternative Reference Rates) indices, also called RFR. Both indices try to represent the same thing, the risk-free rate, but they are fundamentally different. While the former represents the average rate at which Panel Banks believe they could borrow money, the latter is calculated backward based on transactions. Therefore, these new products need their corresponding convexity adjustment. 

The first references on the convexity adjustment were \cite{RitchkenS}, \cite{Flesaker} and \cite{BrothertonIben}, published almost simultaneously. A convexity formula for averaging contracts was found in \cite{RitchkenS}. Flesaker derived a convexity adjustment for computing the expected Libor rate under the Ho-Lee model in a continuous and discrete setting in \cite{Flesaker}. \cite{BrothertonIben} used the Taylor expansion on the inverse function for calculating the convexity adjustment. In the following years, several improvements were made. For example, the convexity adjustment was extended to other payoffs in \cite{Hull06}. \cite{Hart} improved the Taylor expansion. \cite{KirikosNovak} derived the convexity adjustment for the Hull-White model. Afterwards, we can find papers that extend the convexity adjustment to different payoffs, see \cite{Benhamou00WC} or \cite{Hagan03}. Or by applying alternative techniques such as the change of measure in \cite{Pelsser}, a martingale approach in \cite{Benhamou00} or the effects of stochastic volatility in \cite{PiterbargRenedo} and \cite{HaganWoodward20}.

In the present paper, we find an alternative way to calculate the convexity adjustment for a general interest rate model. The idea is to use the It\^o's representation theorem. Unfortunately, the theorem does not give an insight into how to calculate the elements therein. Therefore, it is necessary to introduce basic concepts of Malliavin calculus to apply the Clark-Ocone representation formula.

The structure of the paper is as follows. In Section \ref{sec:Malliavin}, we give a brief introduction to Malliavin calculus. In Section \ref{sec:Notation}, we provide preliminary information and discuss the notation used in the article. The notation introduced will not be repeated unless necessary to guide the reader through the results. In Section \ref{sec:CA}, we compute the convexity adjustment for several payoffs commonly negotiated in the interest rate trading desks. We also present some numerical experiments to check the accuracy of the analytical results obtained. We provide the code for these examples at \url{https://github.com/Dagalon/GeneralConvexityAdjustment}.
The conclusions can be found in Section \ref{sec:Conclusion}, as well as future lines of research to explore.

\section{Basic introduction to Malliavin calculus}\label{sec:Malliavin}
Malliavin calculus is an infinite-dimensional calculus in a Gaussian space. In other words, this is a theory that provides a way to calculate the derivatives of random variables defined in a Gaussian probability space with respect to the underlying noise. The initial objective of Malliavin was the study of the existence of densities of Wiener functionals such as solutions of stochastic differential equations. But, nowadays, it has become an important tool in stochastic analysis due to the increase in its applications. Some of these applications include stochastic calculus for fractional Brownian motion, central limit theorems for multiple stochastic integrals, and an extension of the Itô formula for anticipative processes, but especially mathematical finance. For example, we can apply Malliavin calculus for computing hedging strategies, Greeks, or obtain price approximations. See, for example, \cite{AlosLorite} or \cite{Nualart} for more general content.\\

In our case, we are interested in using the Malliavin calculus to apply the Clark–Ocone representation theorem. But, first of all, let's introduce some basic concepts.\\ 

Now, we introduce the derivative operator in the Malliavin calculus sense and the divergence operator to establish the notation that we use in the remainder of the paper.\\

Consider $W=\left\{W_{t}, t\in \left[0,T\right]\right\}$ a Brownian motion defined on a complete probability space $\left(\Omega, \mathcal{F}, \mathbb{P}\right)$ such that  $\mathcal{F}$ is generated by $W$, equipped with its Brownian filtration $(\mathcal{F}_{t})_{t\geq0}$. Let $H=L^{2}(\left[0,T\right])$ and denote by 
\begin{eqnarray*}
W(t) := \int^{T}_{0} h(s) \d W_{s},
\end{eqnarray*}
the It\^o integral of a deterministic function $h \in H$, also known as Wiener integral. Let $\mathcal{S}$ be the set of smooth random variables of the form
\begin{eqnarray*}
F=f\left(W(t_{1}), \ldots, W(t_{n})\right)
\end{eqnarray*}
with $t_{1}, \ldots, t_{n} \in \left[0,T\right]$ and $f$ is a infinitely differentiable bounded function.

\medbreak

The derivative of a random variable $F$, $D_{s}F$, is defined as the stochastic process given by
\begin{eqnarray*}
D_{s}F = \sum^{n}_{i=1} \frac{\partial f}{\partial x_{i}}\left(W(t_{1}), \ldots, W(t_{n})\right)1_{\left[0, t_{j}\right]}(s), \hspace{0.4cm} s \in \left[0,T\right].
\end{eqnarray*}

%The iterated derivative operator of a random variable $F$ is defined by
%\begin{eqnarray*}
%D^{m}_{s_{1}, \ldots, s_{m}} F = D_{s_{1}} \cdots D_{s_{m}}F, \hspace{0.4cm} s_{1}, \ldots, s_{m} \in \left[0,T\right].
%\end{eqnarray*}

\cite{Nualart} stated that these operators are closable from $L^{p}(\Omega)$ into $L^{p}(\Omega; L^{2}\left[0,T\right])$ for any $p\geq1$, and we denote by $\mathbb{D}^{n,p}$ the closure of $\mathcal{S}$ with respect to the norm
\begin{eqnarray*}
\left\|F\right\|_{n,p} := \left(E\left[F\right]^{p} + \sum^{n}_{i=1} E\left\|D^{i} F\right\|^{p}_{L^{2}\left(\left[0,T\right]^{i}\right)} \right)^{\frac{1}{p}}.
\end{eqnarray*}

\medbreak 

We define $\delta$ as the adjoint of derivative operator $D$, also referred to as the Skorohod integral. The domain of $\delta$, denoted by $Dom\text{ }\delta$, is the set of elements $u \in L^{2}([0,T] \times \Omega)$ such that there exists $\delta(u) \in L^{2}(\Omega)$ satisfying the duality relation
\begin{eqnarray*}
E\left[\delta(u) F \right] = E\left[\int^{T}_{0} D_{s} F u_{s} \d s\right].
\end{eqnarray*}

The operator $\delta$ is an extension of the Itô integral in the sense that the set $L^{2}_{a}([0,T] \times \Omega)$ of square integrable and adapted processes is included in $Dom\text{ }\delta$ and the operator $\delta$ restricted to $L^{2}_{a}([0,T] \times \Omega)$ coincides with the Itô stochastic integral.

\medbreak

For any $u \in Dom\text{ }\delta$, we will use the following notation
\begin{eqnarray*}
\delta(u)=\int^{T}_{0}u_{s}\d W_{s}.
\end{eqnarray*}

\medbreak

The representation of functionals of Brownian motion by stochastic integrals, also known as martingale representation, has been widely studied over the years. It states that if $F$ is a square-integrable random variable, there exists a unique adapted process $\varphi$ in $L^{2}(\Omega \times \left[0,T\right]; \mathbb{R})$ such that
\begin{eqnarray*}
F=E\left[F\right] +\int^{T}_{0} \varphi^{i}_{s} dW^{i}_{s}.
\end{eqnarray*}
In other words, there exists a unique martingale representation or, more precisely, the integrand $\varphi$ in the representation exists and is unique in $L^{2}(\Omega \times \left[0,T\right]; \mathbb{R})$.

\medbreak

Unfortunately, it is not easy to find an analytic representation of the process $\varphi$. Here, the Malliavin calculus helps us to find a solution. When the random variable $F$ is Malliavin differentiable, the process $\varphi$ appearing in It\^o's representation theorem, is given by
\begin{eqnarray*}
\varphi^{i}=E\left[D^{W}_{s}F|\mathcal{F}_{s}\right].
\end{eqnarray*}
In fact,
\begin{eqnarray}\label{clark-okone}
F=E\left[F\right] + \int^{T}_{0} E\left[D^{W}_{s}F|\mathcal{F}_{s}\right] dW_{s}
\end{eqnarray}
is the Clark-Ocone representation formula.

\section{Preliminaries and notation}\label{sec:Notation}
In this section, we give the basic preliminaries and notation necessary throughout the paper.
\subsection{A tale of two curves}
Consider a continuous-time economy where zero-coupon bonds are traded for all maturities. The price at time $t$ of a zero-coupon bond with maturity $T$ is denoted by $P(t,T)$ where $0\leq t \leq T$. Clearly, $P(T,T)=1$. The compounded instantaneous forward rate is defined as:
\begin{eqnarray*}
f(t,T)= -\partial_{T}\ln P(t,T)
\end{eqnarray*}
and the spot interest rates as:
\begin{eqnarray*}
r(t)=\lim_{T\longrightarrow t} -\partial_{T}\ln P(t,T).
\end{eqnarray*}
Therefore, the zero-coupon bond price is given by
\begin{eqnarray*}
P(t,T)=\exp\left(-\int^{T}_{t} f(t,u) du\right).
\end{eqnarray*}

%Before the financial crisis, there was a single curve framework based on the same curve for discounting and forecasting. Since then, the market has adopted a multi-curve approach with two different curves: the discount curve and the estimation curve chosen based on the maturity of the underlying rate. The difference between these two curves is known as the basis. In this paper, we will assume that the basis are not stochastic. Therefore, it can be obtained directly from the market at time $t=0$. In other words, the estimation forward curve $f_{E}(t, T)$ is given by
%\begin{equation}\label{estimation_forward_rate_curve}
%f_{E}(t,T) = f_{ois}(t,T) + s(t,T)
%\end{equation}
%where and $f_{ois}$ is the discount curve and $s(t,T)$ are the basis between the two curves, i.e. $s(t,T)= f_{E}(0,T) - f_{ois}(0,T)$.\\

Before the financial crisis, there was a single curve framework based on the same curve for discounting and forecasting. Since then, the market has adopted a multi-curve approach with two different curves: the discount curve and the estimation curve. We will use the following notation:
\begin{itemize}
	\item The \textbf{\textit{discount forward curve}} is built with OIS instruments which are considered the best approximation for the risk-free rate. We will denote the forward discount rate curve by $f_{ois}(t,T)$ and the discount curve by $P_{ois}(t,T)$.

	\item The \textbf{\textit{estimation forward curve}} is chosen based on the maturity of the underlying rate. Until the crisis, the spread between the OIS and the Ibor was negligible. For example, the OIS 6M and Ibor 6M. Nowadays, due to credit and liquidity reasons, there is a spread between them. As a consequence, the estimation curve is tenor-dependent. We will denote the forward estimation rate curve by $f_{E}(t,T)$ and the estimation discount curve by $P_{E}(t,T)$. 
	
	\item The \textbf{\textit{basis forward curve}} is the difference between the estimation forward curve and the discount forward curve, i.e. $s(t,T)= f_{E}(t,T) - f_{ois}(t,T)$. In this paper, we will assume that the basis are not stochastic. Therefore, it can be obtained directly from the market at time $t=0$ i.e $s(t,t + u)= s(0,u)$ for $u \geq 0$.
	\end{itemize}

Consequently, the estimation forward curve $f_{E}(t, T)$ is given by
\begin{equation}\label{estimation_forward_rate_curve}
f_{E}(t,T) = f_{ois}(t,T) + s(t,T).
\end{equation}

Given the discount curve $P_{ois}(t,T)$ and using the representation \eqref{estimation_forward_rate_curve}, it is possible to find the discount curve for the estimation curve using the relation
\begin{equation}\label{bond_forward}
P_{E}(t,T)=H(t,T)P_{ois}(t,T)
\end{equation}
where $H(t,T)=\exp\left(-\int_{t}^{T}s(t,u) du \right)$.

\subsection{The model}
We will assume that the $f_{ois}$ dynamics follows a single factor Heath-Jarrow-Morton (HJM) model under the $\mathbb{Q}$-measure. Therefore, let $T>0$ a fixed time horizon, $t>0$ the starting time, and $W$ a Brownian motion defined on a complete probability space $(\Omega, \mathcal{F}, \mathbb{P})$. Then, the HJM model is defined by
\begin{align}\label{ois_forward_rate_curve}
df_{ois}(t,T) &= \sigma(t,T)\nu(t,T)dt + \sigma(t,T)dW^{\mathbb{Q}}_t
\end{align}
where $\nu(t,T)=\int_{t}^{T}\sigma(t,s)ds$ and $\sigma(t, T)$ are $\mathcal{F}_{t}$-adapted process that are positive functions for all $t,T$. In particular, we have that
\begin{eqnarray*}
f_{ois}(t,T)= -\partial_{T}\ln P_{ois}(t,T).
\end{eqnarray*}

To have a Markovian representation of the HJM, we will assume that the volatility is separable, i.e.
\begin{equation}\label{separation_condition}
\sigma(t,T)= h(t)g(T)
\end{equation}
with $g$ a positive time-dependent function and $h$ a non-negative process. This version of the HJM is also known as the Cheyette model, \cite{Cheyette}.\\

In particular, following \cite{AndreasenPiterbarg}, we will define
\begin{align*}
\eta_t &= g(t)h(t,x_t,y_t)  \nonumber \\
k_t &= - \frac{\partial_t g(t)}{g(t)}.
\end{align*}
Then, we have the following proposition.
\begin{proposition}
Consider the HJM model \eqref{ois_forward_rate_curve} with the separable volatility condition \eqref{separation_condition}. Define the stochastic processes $x(t)$ and $y(t)$ by
\begin{align}\label{short_rate_cheyette}
dx_t &= \left(-k(t) x(t) + y(t)\right)dt + \eta\left(t,x(t),y(t)\right) dW_t^{\mathbb{Q}} \nonumber \\
dy_t &= \left(\eta^{2}(t) - 2 k(t) y(t)\right) dt ,\nonumber \\
x(0) &=y(0)=0.
\end{align}
All zero-discount bonds are deterministic functions of the processes $x(t)$ and $y(t)$,
\begin{eqnarray*}
P_{ois}(t,T)=P(t, T, x(t), y(t)), 
\end{eqnarray*}
where 
\begin{eqnarray}\label{bond_ois}
P_{ois}(t,T,x,y) = \frac{P_{ois}(0,T)}{P_{ois}(0,t)} \exp\left(-G(t,T)x - \frac{1}{2} G^{2}(t,T)y \right),
\end{eqnarray}
where $G(t,T) = \int_{t}^{T} \exp\left(-\int_{t}^{u} k(s) ds \right) du$ and the short rate is
\begin{equation}
r_{ois}(t)=f_{ois}(t,t)= f_{ois}(0,t) + x(t).
\end{equation} 
\end{proposition}

The whole interest rate curve can be reduced to the evolution of the two-state variables $x(t)$ and $y(t)$. The variable $x(t)$ constitutes the main yield curve driver, whereas $y(t)$ is an auxiliary `convexity' variable. Note that the function $y(t)$ is not deterministic, however, it does not have a diffusion term. We call such processes locally deterministic. \\

We can see from \eqref{short_rate_cheyette} that
\begin{equation*}
x(t_a) = \int_{0}^{t_a} \exp\left(-\int_{u}^{t_a}k(w) dw\right) y(u) du + \int_{0}^{t_a}  \exp\left(-\int_{u}^{t_a}k(w) dw \right) \eta(u,x(u),y(u)) dW_u^{\mathbb{Q}}. 
\end{equation*}

In order to have a more manageable model, we will follow the ideas of \cite{AndreasenPiterbarg} where the state variables are approximated. Although it is possible to use other methodologies, for example \cite{Gatarek2019TowardsAG}. So, we can approximate $y(t)$ as
\begin{equation}\label{approximation_y_t}
\bar{y}(t):=\int_{0}^{t} \exp\left(-2\int_{u}^{t} k(w) dw \right) \eta^{2}(u,x(0),y(0)) du
\end{equation} 
and
\begin{eqnarray}
\bar{x}_{t_a}:= \bar{x}_0(t_a)&+&\int_{0}^{t_a} \exp\left(-\int_{u}^{t_a}k(w) dw\right) \bar{y}_u du \nonumber \\
&+& \int_{0}^{t_a}  \exp\left(-\int_{u}^{t_a}k(w) dw \right) \eta(u,\bar{x}(u),\bar{y}_u) dW_u^{\mathbb{Q}}   \label{approximation_x_t_a}
\end{eqnarray}
with initial condition $\bar{x}_{0}(t_a)$ will be chosen appropriately depending on the case. In the appendix \ref{estimation_error_l2} we can find the estimation for $\mathbb{E}\left[(x_t-\bar{x}_t)^{2}\right]$ and the approximation order.

\subsection{Model constraints}
To calculate the convergence order of the convexity adjustment approximation, we use the following hypotheses on $\eta(t,x,y)$.
\begin{hypothesis}\label{boundedness_volatility} 
The process $\eta_t$ is global Lipschitz and differentiable a.s. In addition, we will suppose that
\begin{align*}
\alpha_1 \leq \eta(t,x,y) \leq \alpha_2 \quad \forall (t,x,y) \in \mathbb{R}^{+} \times \mathbb{R} \times \mathbb{R}^{+} \quad \text{and} \quad \alpha_1, \alpha_2 > 0. \\
|\eta(t,x_2,y_2) - \eta(t,x_1,y_1)| \leq C_{x,y} \lVert (x_2-x_1,y_2-y_1)\rVert \quad \forall (t,x,y) \in \mathbb{R}^{+} \times \mathbb{R} \times \mathbb{R}^{+} 
\end{align*}
with $\lVert \cdot \rVert$ euclidean norm in $\mathbb{R}^{2}$.
\end{hypothesis}

The mean reversion function $k(\cdot)$ influences the range and flexibility of the volatility structure. The function is always positive and, in practice, it is usually low.

\begin{hypothesis}\label{boundedness_reversion} 
The mean reversion function $k(\cdot)$ is a continuous and positive a.s such that
\begin{equation*}
m_k < k(t) \leq M_k \quad \forall t \geq 0.
\end{equation*}

As a consequence.
\begin{remark}
Under these assumptions on $k(\cdot)$, we have that
\begin{equation*}
\lim_{t \to \infty}  I(\alpha,0,t):= \lim_{t \to \infty} \int_{0}^{t} \exp\left(-\alpha \int_{u}^{t} k(s) ds\right) du \leq \frac{1}{\alpha m_k} \quad \text{with} \quad \alpha > 0.
\end{equation*}
On other hand, 
\begin{equation*}
\lim_{t \to \infty}  J(\alpha,0,t):= \lim_{t \to \infty} \int_{0}^{t} G^{\alpha}(u,t) \exp\left(-\alpha \int_{u}^{t} k(s) ds\right) du \leq \frac{1}{\alpha m^{\alpha+1}_{k}}
\end{equation*}
\end{remark}
\end{hypothesis}

The hypotheses have been chosen for simplicity, but they can be replaced by suitable integrability conditions. We should also note that under the hypothesis \eqref{boundedness_volatility}, $\partial_x \eta(t,x,y)$ and $\partial_y \eta(t,x,y)$ are bounded.

\section{Convexity Adjustment}\label{sec:CA}
In this section, we derive the convexity adjustment for different products. The advantage of using the Malliavin calculus is that it allows us to derive a general representation formula for the convexity adjustment. Furthermore, as we will see later, it is possible to obtain closed formulas for the convexity adjustment when the volatility of the Cheyette model is time-dependent. 

We are going to introduce a general idea of the method. Let us define the process $Z_t = f(x_t)$. Suppose that $Z_t$ is a martingale under a measure $\mathbb{Q}_1$. However, we are interested in calculate $\mathbb{E}^{\mathbb{Q}_2}\left[Z_T \right]$ where $\mathbb{Q}_2$ is a measure under which $Z_t$ is not martingale and such that $dW^{\mathbb{Q}_1}_t = dW^{\mathbb{Q}_2}_t +\lambda_t dt$. Then, if we use the Clark-Ocone representation, we have that
\begin{equation*}
f(x_t) = \mathbb{E}^{\mathbb{Q}_1}\left[f(x_t)\right] + \int_{0}^{t} \mathbb{E}^{\mathbb{Q}_1}_s\left[ f^{\prime}(x_t) D_s x_t  \right] dW^{\mathbb{Q}_1}_s
\end{equation*}
Now, taking $\mathbb{E}^{\mathbb{Q}_2}\left( \cdot \right)$ in the previous expression and using Girsanov's theorem, we get that
\begin{equation}\label{general_convexity}
\mathbb{E}^{\mathbb{Q}_2}\left[ f(x_t) \right] = f(x_0) + \mathbb{E}^{\mathbb{Q}_2} \left[\int_{0}^{t}  \mathbb{E}^{\mathbb{Q}_1}_s\left[ f^{\prime}(x_t) D_s x_t  \right] \lambda_s ds \right]. 
\end{equation}
The second term is the convexity adjustment due to the change of measure from $\mathbb{Q}_1$ to $\mathbb{Q}_2$. The different choices of $f$, $\mathbb{Q}_1$, and $\mathbb{Q}_2$ will allow us to obtain a convexity adjustment approximation for the different cases of interest. 
\subsection{FRAs Vs futures}
The cash flows in FRAs and futures are computed under different measures. Consequently, we need to adjust the futures price quote to transform them into FRAs price quotes. On one hand, we define the forward rate at time $t_0$ between $t_1$ and $t_2$ under the forward curve $E$ as:
\begin{equation}\label{forward_rate}
L_{E}(t_0, t_1, t_2) = \frac{1}{\delta_{t_1,t_2}}\left(\frac{P_{E}(t_0,t_1)}{P_{E}(t_0,t_2)} - 1 \right)
\end{equation} 
where $P_{E}(t,T)$ is the discount factor for the curve $E$ from $t$ to $T$, and $\delta_{t_1,t_2}$ is the year fraction between $t_1$ and $t_2$ and, $t_{0}\leq t_{1} \leq t_{2}$. Note that $L_{E}(t_0, t_1, t_2)$ is a martingale under the forward measure $\mathbb{Q}^{t_2}$ associated with the numeraire $P_{ois}(t,t_2)$.\\

On the other hand, given $t\leq t_{0}$, let us define the future rate as:    
\begin{equation}\label{future}
\hat{L}_{E}(t,t_0, t_1, t_2) = \mathbb{E}_t^{\mathbb{Q}}\left[L_{E}(t_0, t_1, t_2) \right]
\end{equation}
where $\mathbb{Q}$ is the measure associated to the numeraire $B_t=\exp\left(\int_{0}^{t} r_{ois}(s) ds \right)$ with $ r_{ois}(t)$ the risk free short rate. Using
\eqref{forward_rate} and \eqref{future}, then the convexity adjustment definition is
\begin{equation*}
CA(t, t_0, t_1, t_2) = \hat{L}_{E}(t,t_0, t_1, t_2) - \mathbb{E}_t^{\mathbb{Q}^{t_2}}\left[L_{E}(t_0, t_1, t_2) \right].
\end{equation*}

In the following theorem, we specify the convexity adjustment for the futures.

\begin{theorem}\label{Th_CA_futures}[Convexity Adjustment approximation for Futures]
Given the Cheyette model in \eqref{short_rate_cheyette}, the hypotheses \ref{boundedness_volatility} and \ref{boundedness_reversion}, and considering the approximations in \eqref{approximation_y_t} and \eqref{approximation_x_t_a}. Then, the convexity adjustment approximation for futures is 
\begin{equation}\label{ca_approximation_futures}
CA(t,t_0,t_1, t_2) = \frac{P_{E}(0,t_1)}{\delta_{t_1,t_2} P_{E}(0,t_2)} \biggl(G(t_0,t_2)  - G(t_0,t_1) \biggr) \int_{0}^{t_0} \beta(s,t_0,\bar{x}_0(t_0),\bar{y}_s) \nu(s,t_2) ds + E(t_0)
\end{equation}
with 
\begin{equation*}
\beta(u,t_0,x,y) = \exp\left(-\int_{u}^{t_0}k(w) dw \right) \eta(u,x,y).
\end{equation*}
The error $E(t_0)$ is given by (\ref{f_x_y}) with $f(x,y)=\frac{1}{\delta_{t_0,t_1}}\left(\frac{P_{ois}(t_0,t_1,x,y)}{P_{ois}(t_0,t_2,x,y)} -1\right)$ and behaves as $\mathcal{O}(t_0)$ when $t_0 \to 0$ and $\|E(t_0)\|^{2}_{2} < \infty$ when $t_0 \to \infty$. 
\end{theorem}
\begin{proof}
See appendix \ref{Proof_CA_futures}.
\end{proof}

\begin{example}[Convexity adjustment for futures under the Hull-White model]\label{example_ca_future}
The Cheyette model can be reduced to Hull-White model using the following parameters
\begin{align*}
g(T) &= \exp(-kT), \\
h(t) &= \exp(k t) \sigma.
\end{align*}
Moreover, from the definition of $g(\cdot)$ and $h(\cdot)$, we have that
\begin{align*}
\eta_s &= \sigma ,\\
\beta(s,u, x_0, \bar{y}_s) &= \sigma \exp(-k(u-s)),\\
\nu(s,t_2) &= \sigma \frac{1 - \exp(-k(t_2-s))}{k}.
\end{align*}
Then the convexity adjustment \eqref{ca_approximation_futures} is
\begin{align*}
CA(t_0,t_1, t_2) & \approx \frac{\sigma^{2} \exp(-k t_0)  P_{E}(0,t_1)}{\delta_{t_1,t_2} P_{E}(0,t_2)} \left(\frac{1 - \exp(- k t_0)}{k^{2}} - \frac{t_0 \exp(-k t_2)}{k} \right).   
\end{align*}

In the Figure \ref{fig:Futures}, we can check the accuracy of the last formula versus Monte Carlo simulation. The parameters used are $\sigma=0.015$, $k=0.003$, and flat curve with level $r=0.01$.

\begin{figure}[H]
	\begin{center}
		\includegraphics[scale=0.3]{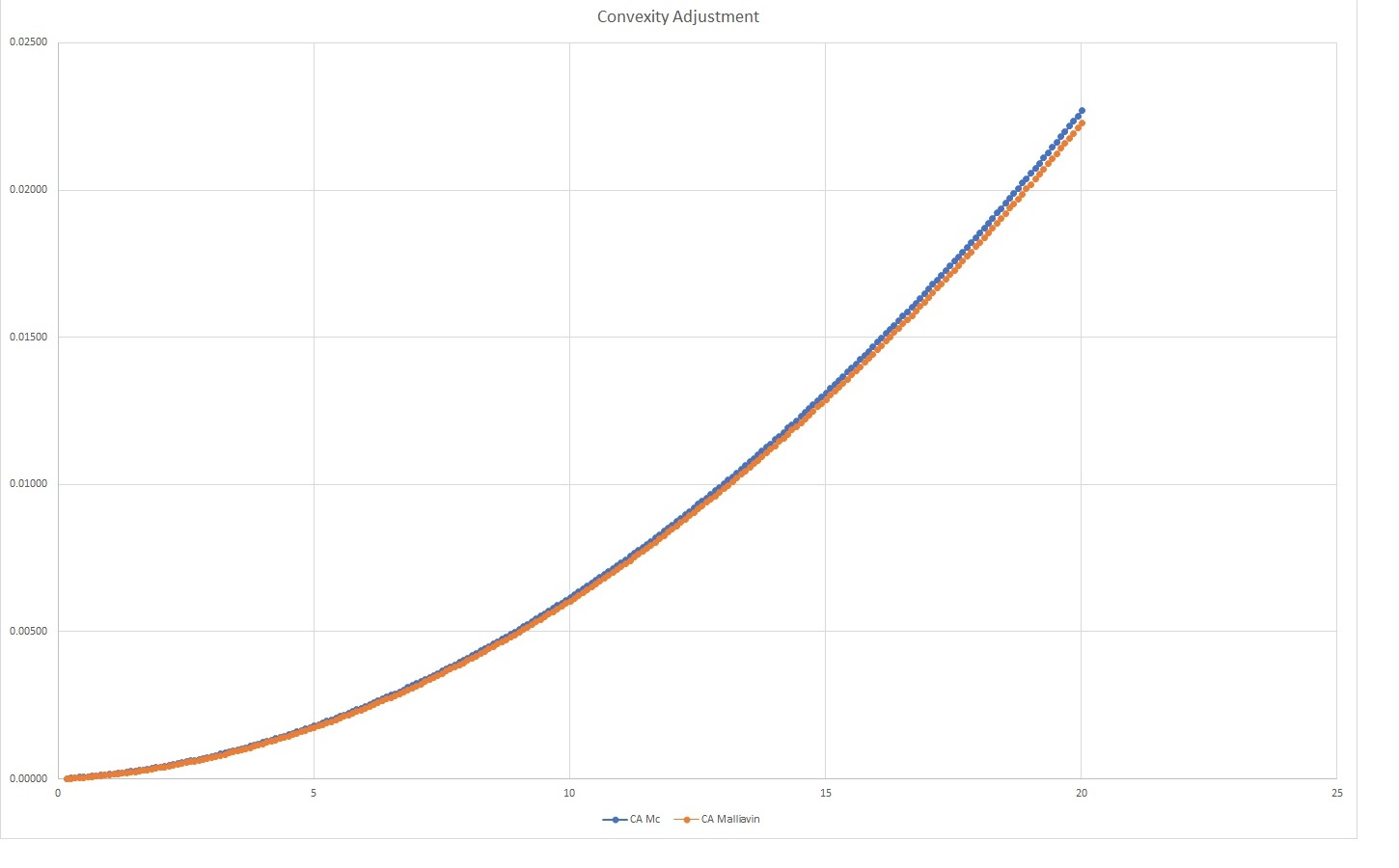}
	\end{center}
	\caption{Futures: Comparison Malliavin vs MC Simulation}
	\label{fig:Futures}
\end{figure} 
\end{example}

\subsection{OIS futures}
In this section, we will derive the convexity adjustment for short-term interest rate swaps, also known as STIRs. It is common in the market to find two versions of these futures, depending on how the fixings of the reference index are used. Given $t_0 < t_1$, we will define the overnight compounding rate as
\begin{equation*}
R(t_0,t_1) :=\frac{1}{\delta_{t_0,t_1}} \left(\exp\left(\int_{t_0}^{t_1}r_{ois}(u) du \right) - 1\right)
\end{equation*}
and the overnight average rate as
\begin{equation*}
R_{avg}(t_0,t_1) := \frac{1}{\delta_{t_0,t_1}}\int_{t_0}^{t_1}r_{ois}(u) du.
\end{equation*}
We observe that both, $R(\cdot,t_0,t_1)$ and $R_{avg}(\cdot,t_0,t_1)$  are not predictable and are only observable in $t_1$. However, 
$R(\cdot,t_0,t_1)$ and $R_{avg}(\cdot,t_0,t_1)$ are flows that will be payed in $t_1$. Therefore, we can consider that the expected value under the measure $\mathbb{Q}$ is observable during the entire period $[t_0, t_1]$. Let us define the next $\mathbb{Q}$ martingales:
\begin{align*} 
\bar{R}(t,t_0,t_1) &:= \mathbb{E}_t^{\mathbb{Q}}\left[ R(t_0,t_1)  \right], \\
\bar{R}_{avg}(t,t_0,t_1) &:= \mathbb{E}_t^{\mathbb{Q}}\left[ R_{avg}(t_0,t_1)  \right].
\end{align*}
Before continuing, we will do several observations. The first observation is that if we define $F(t,t_0,t_1) = \mathbb{E}^{\mathbb{Q}^{t_1}}\left[ R(t_0,t_1)\right]$, then we have that when $t \in [0,t_0]$
\begin{align*}
F(t,t_0,t_1)&= \frac{1}{P_{ois}(t,t_1)}  \mathbb{E}_{t}^{\mathbb{Q}}\left[\exp\left(-\int_{t}^{t_1} r_{ois}(u) du \right) R(t_0,t_1) \right]\\
&= \frac{1}{\delta_{t_0,t_1}}\left(\frac{P_{ois}(t,t_0)}{P_{ois}(t,t_1)} - 1\right), %, \quad t \in [0,t_0] \\
%\end{align*}
\intertext{and when $t \in  [t_0, t_1]$, we have}
%\begin{align*}
F(t,t_0,t_1)&= \frac{1}{P_{ois}(t,t_1)} \mathbb{E}_{t}^{\mathbb{Q}}\left[\exp\left(-\int_{t}^{t_1} r_{ois}(u) du \right) R(t_0,t_1) \right]\\
 &= \frac{1}{\delta_{t_0,t_1}} \left(\frac{\exp\left(\int_{t}^{t_1}r_{ois}(u) du\right)}{P_{ois}(t_0,t)}-1\right).%, \quad  t \in  [t_0, t_1].
\end{align*}
Then, the convexity adjustment for $ R(t_0,t_1)$ is
\begin{equation}\label{R_ois_ca}
CA_{ois}(t,t_0,t_1) = \bar{R}(t,t_0,t_1) - F(t,t_0,t_1).
\end{equation}
The second observation, is that we have the following equality
\begin{equation}\label{R_ois_avg}
\mathbb{E}^{\mathbb{Q}}\left[R_{avg}(t_0,t_1) \right] =\frac{1}{\delta_{t_0,t_1}} \mathbb{E}^{\mathbb{Q}}\left[ \log\left(1+ \delta_{t_0,t_1} R(t_0,t_1) \right) \right]  \end{equation}
To avoid complexity with the notation, we will define 
\begin{equation*}
I(t_0,t_1) := \int_{t_0}^{t_1} r_{ois}(s) ds.
\end{equation*}

\begin{theorem}\label{Th_CA_OIS}[Convexity Adjustment approximation for OIS Futures]
Given the Cheyette model in \eqref{short_rate_cheyette}, the hypotheses \ref{boundedness_volatility} and \ref{boundedness_reversion}, and considering the approximations in \eqref{approximation_y_t} and \eqref{approximation_x_t_a}. Then, the convexity adjustment approximation for OIS futures is 
\begin{equation}\label{convexity_ois_future}
\bar{R}(t_0,t_1) = \frac{1}{\delta_{t_0,t_1}}\left(\exp\left(\mathbb{E}_t^{\mathbb{Q}}\left[I(t_0,t_1)\right]\right)\exp\left(-\frac{1}{2}\int_{t}^{t_1}\Gamma^{2}(s,t_0,t_1) ds\right) - 1\right) + E(t_0),
\end{equation}
and for the average OIS future is
\begin{equation}\label{convexity_avg_ois_future}
\bar{R}_{avg}(t_0,t_1) = \frac{\mathbb{E}_t^{\mathbb{Q}}\left[I(t_0,t_1)\right] }{\delta_{t_0,t_1}} \approx \frac{1}{\delta_{t_0,t_1}}\left(\log\left(1+\delta_{t_0,t_1}  \mathbb{E}_t^{\mathbb{Q}}\left[R(t_0,t_1)\right] \right) - \frac{1}{2}\int_{0}^{t_1}  \Gamma^{2}(s,t_0,t_1) ds\right),
\end{equation}
where 
\begin{equation}
\Gamma(s,t_0,t_1)= \eta(s,x_0,y_0)\int_{\max(s, t_{0})}^{t_1} \exp\left( -\int_{s}^{u} k(w) dw\right)du.
\end{equation}
The error $E(t_0)$ is given by (\ref{f_x_y}) with $f(x,y)=\frac{1}{\delta_{t_0,t_1}}\left(\frac{1}{P_{ois}(t_0,t_1,x,y)} -1\right)$ and behaves as $\mathcal{O}(t_0)$ when $t_0 \to 0$ and $\|E(t_0)\|^{2}_{2} < \infty$ when $t_0 \to \infty$.
\end{theorem}
\begin{proof}
See appendix \ref{Proof_CA_OIS_futures}.
\end{proof}

\begin{remark}
We must note that (\ref{convexity_ois_future}) and (\ref{convexity_avg_ois_future}) are exact when $\eta(t,x_t, y_t)$ is only a time-dependent function, as in the case of the Hull-White model.
\end{remark}

\begin{remark}
We can calculate the convexity adjustment for the case $t_0 < t < t_1$ similarly to when $t < t_0$. For this, we will define
\begin{align*}
I(t,t_1)&:=\int_{t}^{t_1} r_{ois}(s) ds,\\
R(t_0,t_1) &:= \frac{1}{\delta_{t_0,t_1}}\left(\frac{\exp\left(\int_{t}^{t_1} r_{ois}(s) ds\right)}{P_{ois}(t_0,t)} - 1\right), \\
\intertext{and}
R_{avg}(t_0,t_1) &:= \frac{1}{\delta_{t_0,t_1}}\left(\int_{t_0}^{t} r_{ois}(s) ds + \int_{t}^{t_1} r_{ois}(s) ds\right).   
\end{align*}
\end{remark}

\begin{example}[Convexity adjustment for OIS futures under the Hull-White model]\label{example_convexity_hw_ois}
Similarly to the Example \ref{example_ca_future}, we can find the equivalent parameters for the Hull-White model: 
\begin{align*}
\Gamma(s,t_0,t_1) &= \frac{\sigma \exp(-ks)}{k}\biggl(\exp(-k(\max(s,t_0) - s)) - \exp(-k(t_1-s))\biggr),\\
\mathbb{E}^{\mathbb{Q}}\left[I(t_0,t_1)\right]&=-\log\left(\frac{P_{ois}(0,t_1)}{P_{ois}(0,t_0)}\right)\\
& \text{ }\text{ }\text{ }+ \frac{\sigma^{2}}{2k^{2}}\left(\delta_{t_0,t_1} - 2 \frac{\exp(-kt_0) - \exp(-kt_1)}{k} + \frac{\exp(-2kt_0) - \exp(-2kt_1)}{2k}  \right).
\end{align*}
Therefore, we have that
%\begin{align*}
%\frac{\int_{0}^{t_1} \Gamma^{2}(s,t_0,t_1) ds}{2} &= \frac{\sigma^{2}}{2k^2} \int_{0}^{t_1}  \exp(-2ks)\left(\exp(-k(\max(s,t_0) - s)) - \exp(-k(t_1 - s))\right)^{2} ds \\
%&= \frac{\sigma^{2}}{2k^2} \int_{0}^{t_0} \exp(-2ks)\left(\exp(-k(t_0 - s)) - \exp(-k(t_1 - s))\right)^{2} ds\\
%&+ \frac{\sigma^{2}}{2k^2} \int_{t_0}^{t_1} \exp(-2ks)\left(1 - \exp(-k(t_1 - s))\right)^{2} ds\\
%&= \frac{\sigma^{2}t_0}{2k^{2}} \left( \exp(-kt_0) + \exp(-2kt_1) - 2 \exp(-k(t_1+t_0)) \right)\\  
%&+ \frac{\sigma^{2}}{2k^{2}} \left(\frac{\exp(-2kt_0) - \exp(-2kt_1)}{2k}  + \exp(-kt0)t0 - 2 \frac{\exp(-2kt_0) - \exp(-k(t_0 + t_1))}{k}  \right)
%\end{align*}
\begin{align*}
\frac{1}{2}\int_{0}^{t_1} \Gamma^{2}(s,t_0,t_1) ds &= \frac{\sigma^{2}t_0}{2k^{2}} \biggl( \exp(-kt_0) + \exp(-2kt_1) - 2 \exp(-k(t_1+t_0)) \biggr)\\  
&+ \frac{\sigma^{2}}{2k^{2}} \biggl(\frac{\exp(-2kt_0) - \exp(-2kt_1)}{2k}  + \exp(-kt_0)t_0 \\
&\quad - 2 \frac{\exp(-2kt_0) - \exp(-k(t_0 + t_1))}{k}  \biggr).
\end{align*}
Then, if we substitute the last equalities in (\ref{convexity_ois_future}), we get an approximation for OIS future at $t=0$.\\

The following figures show the accuracy of (\ref{convexity_ois_future}) and (\ref{convexity_avg_ois_future}). The parameters used to run the Monte Carlo have been $k=0.003$, $\sigma=0.01$, and flat curve $r=0.01$.

\begin{figure}[H]
	\begin{center}
		\includegraphics[scale=0.3]{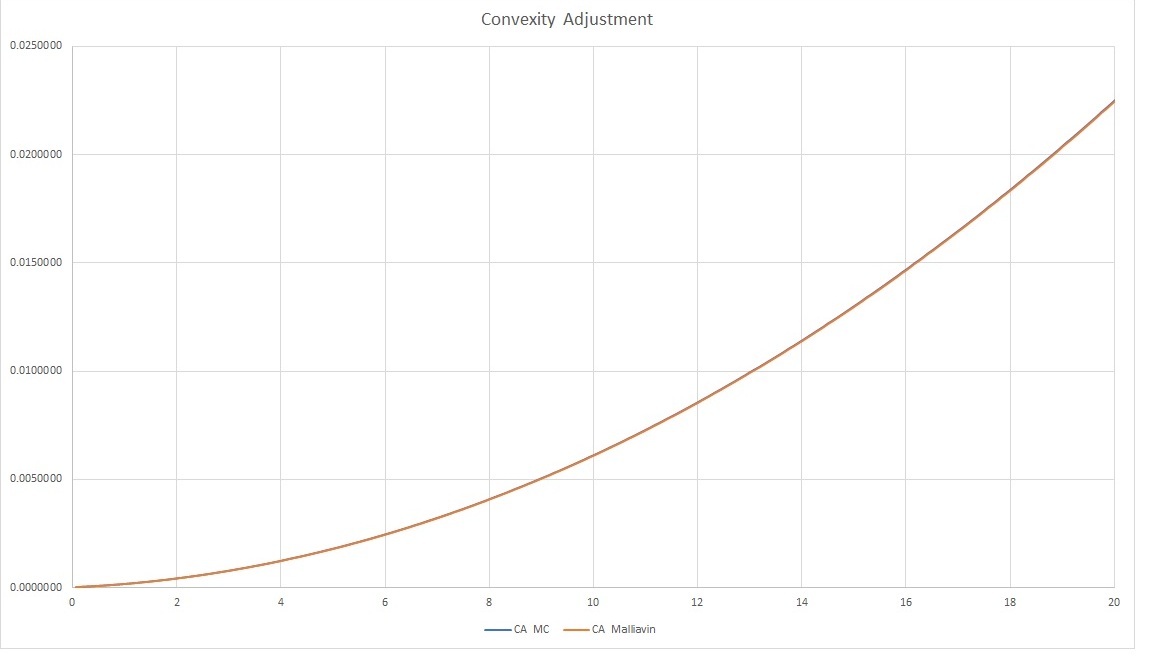}
	\end{center}
	\caption{Compounding OIS Future: Comparison Malliavin vs MC Simulation}
	\label{fig:Futures}
\end{figure}

\begin{figure}[H]
	\begin{center}
		\includegraphics[scale=0.3]{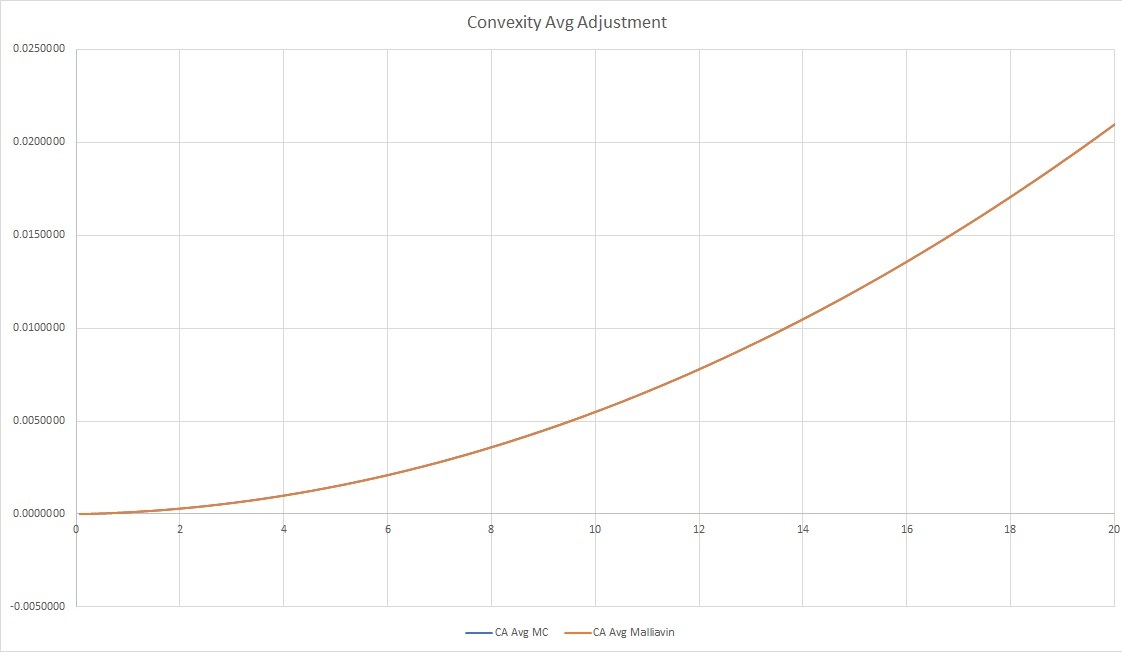}
	\end{center}
	\caption{Average OIS Future: Comparison Malliavin vs MC Simulation}
	\label{fig:Futures}
\end{figure} 
\end{example}

\subsection{FRAs in arrears}
A FRA in arrears is the most classic example among convexity adjustment products. The price is given by
\begin{equation}\label{FRAinArrear}
P_{E}(0,t_1)\mathbb{E}^{\mathbb{Q}^{t_1}}\left[L_{E}(t_1,t_1,t_2)\right],
\end{equation}
i.e. the cash flow associated with a FRA in arrears is $L_{E}(t_1,t_1,t_2)$ in $t_1$.

\begin{theorem}\label{Th_CA_FRAsinArrears}[Convexity Adjustment approximation for FRAs in Arrears]
Given the Cheyette model in \eqref{short_rate_cheyette}, the hypotheses \ref{boundedness_volatility} and \ref{boundedness_reversion}, and considering the approximations in \eqref{approximation_y_t} and \eqref{approximation_x_t_a}. Then, the convexity adjustment approximation for FRAs in Arrears is 
\begin{align}\label{ca_approximations_fra_arrears}
CA(t_0,t_1) =   &\frac{G(t_1,t_2)}{\delta_{t_1,t_2}P_{E}(0,t_1,t_2)}  \nonumber \\
&\cdot \int_{0}^{t_1} \beta(s,t_1, \bar{x}_0(t_1), \bar{y}_s) \overline{DM}(s,t_1) \Bigl(\bar{\nu}(s,t_2,\bar{x}_{0}(t_1))-\bar{\nu}(s,t_1,\bar{x}_{0}(t_1))\Bigr) ds  + E(t_1).
\end{align}
The error $E(t_1)$ is given by (\ref{f_x_y}) with $f(x,y)=\frac{1}{\delta_{t_1,t_2}}\left(\frac{1}{P_{E}(t_1,t_2,x,y)} - 1 \right)$ and behaves as $\mathcal{O}(t_1)$ when $t_1 \to 0$ and $\|E(t_1)\|^{2}_{2} < \infty$ when $t_1 \to \infty$. 
\end{theorem}
\begin{proof}
See appendix \ref{Proof_CA_FRAsinArrears}.
\end{proof}

\begin{example}[Convexity adjustment for FRAs in Arrears under the Hull-White model]\label{example_convexity_hw_FRAsinArrears}
The model can be restricted to a Hull-White model with constant parameters. The analytical approximation obtained from \eqref{ca_approximations_fra_arrears} is
\begin{equation*}
CA(t_0,t_1) \approx \frac{G(t_1,t_2)}{\delta_{t_1,t_2}P_{E}(0,t_1,t_2)}  \frac{\sigma^{2}}{k} \int_{0}^{t_1} \Bigl(\exp(- k(t_1 - u)) -   \exp(- k(t_2 - u))\Bigr) \exp(- k(t_1 - u)) du 
\end{equation*}
In Figure \ref{fig:FRA_HW}, we compare the approximation with a Monte Carlo method when the parameters are $\sigma=0.1$ and $k=0.007$.
\begin{figure}[H]
		\begin{center}
		\includegraphics[scale=0.3]{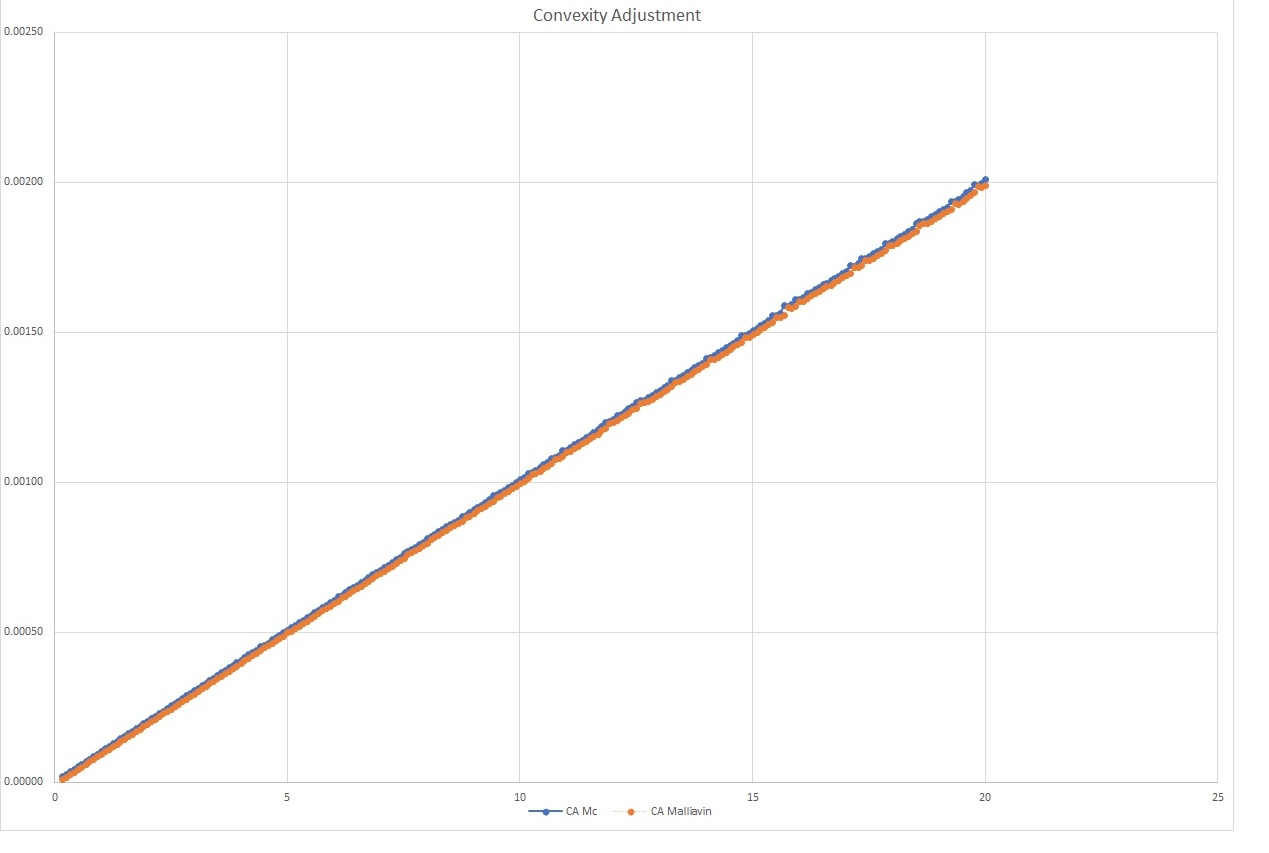}
		\end{center}
		\caption{FRA in Arrears: Comparison Malliavin vs MC Simulation}
		\label{fig:FRA_HW}
\end{figure}
\end{example}

\subsection{CMSs}
The last product we will approximate the convexity adjustment are CMS. We will introduce some notation that we will use throughout the section. We define the swap rate from $t_a$ to $T_b$ at time $t$ as
\begin{equation*}
S_{a,b}(t) := \frac{\sum_{i=1}^{n_E}\delta_{t^{E}_{i-1}, t^{E}_i} L^{E}(t,t^{E}_{i-1}, t^{E}_{i}) P_{ois}(t,t^{E}_{i})}{01(t,t_a,T_b)}
\end{equation*}
where
\begin{align*}
01(t,t_a,t_b) = \sum_{j=1}^{n_f} \delta_{t^{f}_{i-1}, t^{f}_i} P_{ois}(t,t^{f}_{j}) \\
t_a=t^{E}_0 < t^{E}_i< \cdots < t^{E}_{n_E}=t_b \quad i=0,\cdots,n_E&  \\
t_a=t^{f}_0 < t^{f}_j< \cdots < t^{f}_{n_f}=t_b \quad j=0,\cdots,n_f&
\end{align*}
The same way, we will define the OIS swap rate as
\begin{equation*}
S^{ois}_{a,b}(t) = \frac{P_{ois}(t,T^{E}_a) - P_{ois}(t,T^{E}_b)}{01(t,t_a,t_b)}. 
\end{equation*}

\begin{remark}
Note from \eqref{bond_forward} that
\begin{equation*}
S_{a,b}(t) = S^{ois}_{a,b}(t) + \frac{\sum_{i=1}^{n_E}\delta_{t^{E}_{i-1}, t^{E}_i} \alpha(t,t^{E}_{i-1}, t^{E}_{i}) P_{ois}(t,t^{E}_{i})} {01(t,t_a,t_b)}
\end{equation*}
where 
\begin{equation*}
\alpha(t,t^{E}_{i-1}, t^{E}_{i})  = \frac{1}{\delta_{t^{E}_{i-1}, t^{E}_i}}\left(\frac{H(t,t^{E}_{i-1})}{H(t,t^{E}_{i})} - 1\right).
\end{equation*}

We will suppose that variability of spread term structure $\alpha(t,t^{E}_{i-1}, t^{E}_{i})$ is low. Therefore, it is reasonable to freeze it at time $t=0$. Then, we have that
\begin{equation}\label{approximation_basis_swap}
S_{a,b}(t) \approx S^{ois}_{a,b}(t) + \frac{\sum_{i=1}^{n_E}\delta_{t^{E}_{i-1}, t^{E}_i} \alpha(0,t^{E}_{i-1}, t^{E}_{i}) P_{ois}(0,t^{E}_{i})} {01(0,t_a,t_b)}.
\end{equation}
\end{remark}

\begin{theorem}\label{Th_CA_CMS}[Convexity Adjustment approximation for CMS]
Given the Cheyette model in \eqref{short_rate_cheyette}, the hypotheses \ref{boundedness_volatility} and \ref{boundedness_reversion}, and considering the approximations in \eqref{approximation_y_t} and \eqref{approximation_x_t_a}, and
\begin{equation}
M(t,t_p)= \frac{P_{ois}(t,t_p)}{01(t,t_a,t_p)}.
\end{equation}
Then, we have the temporal convexity adjustment for a CMS rate is approximated by
\begin{align} \label{cms_first_order_convexity}
CA(t_p) &\approx  \frac{\partial_x S_{a,b}(t_a,\bar{x}_0(t_a), \bar{y}_{t_a}) \partial_x M(t_a,t_p,\bar{x}_0(t_a), \bar{y}_{t_a})}{M(0,t_p)} \int_{0}^{t_a}  \beta^2(s,t_a,\bar{x}_0(t_a),\bar{y}_s) \\
&\cdot  \exp\left(-2\int_{s}^{t_a}\partial_x (\beta(u,t_a,\bar{x}_u,\bar{y}_u) \mu(u,\bar{x}_u, \bar{y}_u,t_a,t_b))|_{\bar{x}_u=\bar{x}_{0}(t_a)}  du \right)ds + E(t_a)\nonumber 
\end{align}
with  $\beta(u,t_a,x,y) = \exp\left(-\int_{u}^{t_a}k_w dw\right)\eta(u,x,y)$ and the error $E(t_a)$ is given by (\ref{f_x_y}) with $f(x,y)=M(t_a,t_p,x,y)S_{a,b}(t_a,x,y)$ and behaves as $\mathcal{O}(t_a)$ when $t_a \to 0$ and $\|E(t_a)\|^{2}_{2} < \infty$ when $t_a \to \infty$.
\end{theorem}
\begin{proof}
See appendix \ref{Proof_CA_CMS}.
\end{proof}

\begin{remark}
The key point is to calculate an approximation of $\mathbb{E}_s^{0,1}\left[ D_s \bar{x}_{t_a}\right]$. The simplest cases are on the Hull-White or Ho-Lee model. The general case is treated in (\ref{approximation_under_annuity_measeure_d_s}). 
\end{remark}

\begin{example}[Convexity adjustment for CMS under the Hull-White model]
To check the accuracy of the last approximation, we compute with a Monte Carlo simulation the exact value of $\mathbb{E}^{t_p}\left[S^{ois}_{a,b}(t_a)\right]$ under spot measure $\mathbb{Q}$, i.e we will compute  $\frac{1}{P_{ois}(0,t_a)}  \mathbb{E}^{\mathbb{Q}}\left[\frac{S_{a,b}(t_a) P_{ois}(t_a,t_p)}{\beta_{t_a}} \right]$. For the Hull-White model case, we have that
\begin{equation*}
D_s x_{t_a} = \sigma \exp(-(t_a - s)),
\end{equation*}
and the volatility is only time-dependent i.e 
$$
\partial_x (\beta(u,t_a,\bar{x}_u,\bar{y}_u) \mu(u,\bar{x}_u, \bar{y}_u,t_a,t_b)) = 0.
$$
Therefore, \eqref{cms_first_order_convexity} is equal to
\begin{equation*}
\mathbb{E}^{t_p}\left(S_{a,b}(t_a)\right) \approx  S^{ois}_{a,b}(0) + \frac{\partial_x S^{ois}_{a,b}(t_a, \bar{x}_0(t_a),\bar{y}_{t_a})\partial_x M(t_a,t_p, \bar{x}_0(t_a),\bar{y}_{t_a})}{M(0,t_p)} \frac{\sigma^{2}(1-\exp(-2kt_a))}{2k}.
\end{equation*}

In Figure \ref{fig:CMS}, we can observe the CMS convexity adjustment when the tenor of the underlying swap is 5Y. We have compared the above approximation and a Monte Carlo simulation for a Hull-White model with parameters $\sigma=0.01$ and $k=0.0007$

\begin{figure}[H]
	\begin{center}
		\includegraphics[scale=0.25]{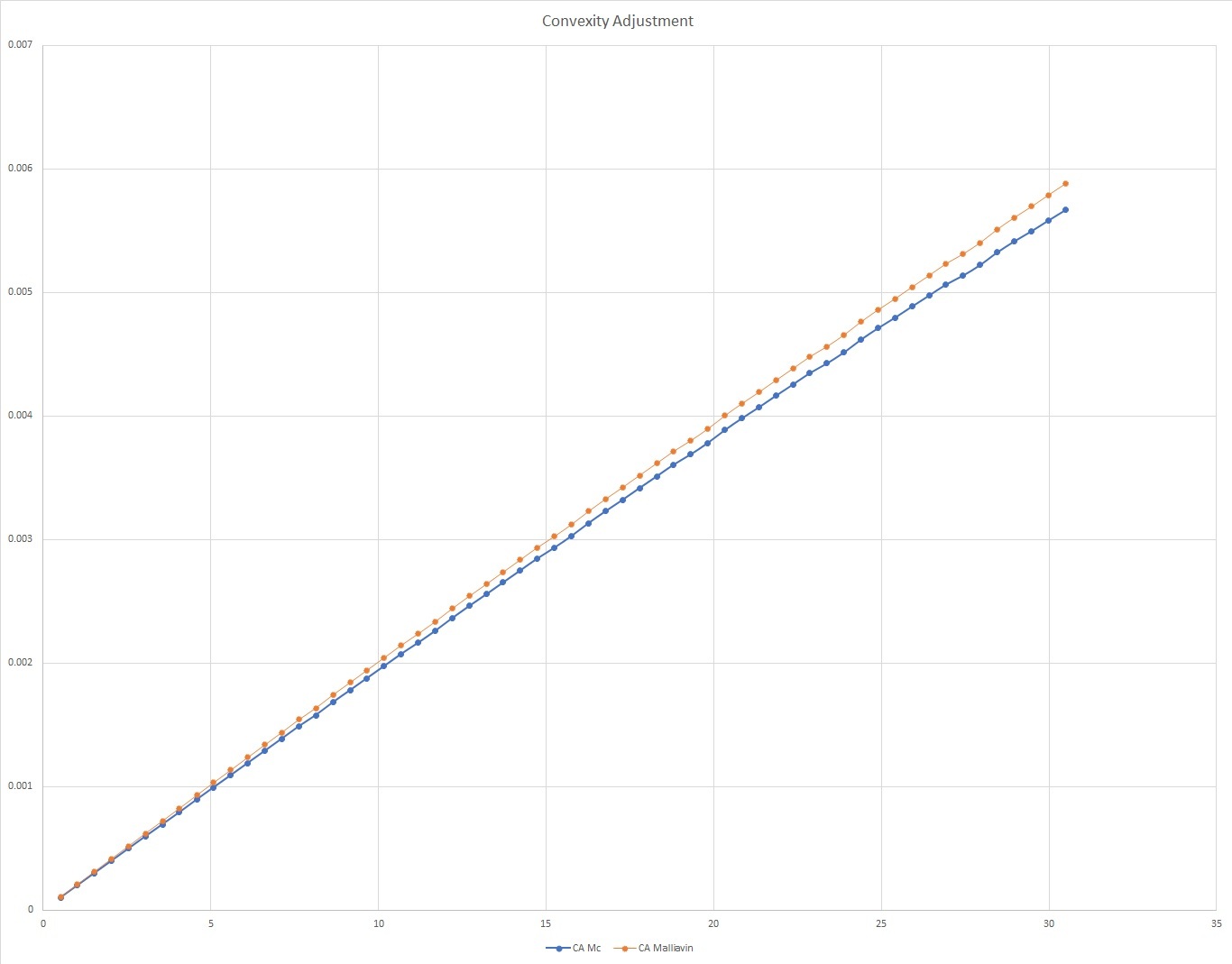}
	\end{center}
	\caption{CMS: Comparison Malliavin vs MC Simulation}
	\label{fig:CMS}
\end{figure} 
\end{example}

\section{Conclusions}\label{sec:Conclusion}
In the present paper, we develop a novel method based on the Malliavin calculus to find an approximation for the convexity adjustment for various classical interest rate products. Although the convexity adjustment could be calculated using other techniques, Malliavin calculus provides a simple way to get a template for the convexity adjustment. We find the approximation for Futures, OIS Futures, FRAs, and CMSs under a general family of the one-factor Cheyette model. We have also seen the excellent quality of the numerical accuracy of the formulas obtained.\\

In future work, the methodology could be extrapolated to a Cheyette model with stochastic volatility or even to a model with several factors.

\section*{Appendix}
\appendix
\renewcommand{\thesection}{\Alph{section}.\arabic{section}}

\section{Auxiliary lemmas}
\subsection{Estimation of $\mathbb{E}_s^{\mathbb{Q}}\left[(x_t - \bar{x}_{t})^{2}\right]$}\label{estimation_error_l2}
To obtain the order of the approximation, we will use the ideas of the paper \cite{BGM}. Basically, we will define a parametrization of the processes $x_{t,\epsilon}$ and $y_{t,\epsilon}$ with $\epsilon \in [0,1]$. The role of $\epsilon$ is only as an interpolation parameter between the process. We will suppose the next Cheyette model parametrization
\begin{align}\label{parametric_process}
y_{\epsilon,t} &= \int_{0}^{t} \exp\left(-2 \int_{u}^{t} k(w) dw\right) \epsilon^{2} \eta^{2}(u,x_{\epsilon,u},y_{\epsilon,u}) du \\
x_{\epsilon, t} &= \int_{0}^{t} \exp\left(-\int_{u}^{t} k(w) dw\right) \bar{y}_{\epsilon,u} du + \int_{0}^{t} \exp\left(-\int_{u}^{t} k(w) dw\right) \epsilon \eta(u,x_{\epsilon,u},y_{\epsilon,u})  \bar{y}_u du
\end{align}
The reason of $\epsilon^{2}$ in $y_{\epsilon,t}$ is because $\mathbb{E}^{\mathbb{Q}}\left[y_t\right]= \mathbb{V}ar(x_t)$. Therefore, if the scale of $x_{\epsilon,t}$ is $\epsilon$, then the scale of $y_{\epsilon,t}$ must be $\epsilon^{2}$.\\

We will start with the estimation of $\mathbb{E}^{\mathbb{Q}}\left[(y_t - \bar{y}_t)^{2}\right]$. We will define the next parametrization for $\nu \in [0,1]$
\begin{equation*}
z^{\nu}_{u} = \eta^{2}(u, \nu x_u, \nu y_u).
\end{equation*}
Then, we have that $z^{1}_{u} = \eta(u,x_u,y_u)$ and $z^{0}_{u} = \eta(u,0,0)$. Therefore,
\begin{align*}
\eta^{2}(u, x_u, y_u) &= \eta^{2}(u, 0, 0) + \int_{0}^{1} \partial_{\nu}z^{\nu}_{u} d\nu \\ 
&=  \eta^{2}(u, 0, 0) + 2 \int_{0}^{1} \eta(u, \nu x_u, \nu y_u)   \nabla_{x,y}\eta(u, \nu x_u, \nu y_u) \cdot (x_u, y_u) d\nu.
\end{align*}
Moreover, we have that
\begin{equation*}
y_t = \bar{y}_t + E_{y}\left(t\right)
\end{equation*}
with 
$$
E_{y}\left(t\right)= 2 \int_{0}^{t} \exp\left(- 2\int_{u}^{t} k(w) dw\right)  \int_{0}^{1} \eta(u, \nu x_u, \nu y_u)   \nabla_{x,y}\eta(u, \nu x_u, \nu y_u) \cdot (x_u, y_u) d\nu du
$$
Now, from Cauchy-Schwarz inequality, we obtain that
\begin{align} \label{upper_bound_E_y}
\mathbb{E}^{\mathbb{Q}}\left[(E_{y}\left(t\right))^{2}\right] &\leq 4 \mathbb{E}^{\mathbb{Q}}\left[\int_{0}^{t} \exp\left(-4 \int_{u}^{t} k(w) dw\right)  \left(\int_{0}^{1} \eta(u, \nu x_u, \nu y_u)   \|\nabla_{x,y}\eta(u, \nu x_u, \nu y_u)\|^{2} \|(x_u, y_u)\| d\nu \right)^2 du \right] ^{2}\nonumber \\ 
&\leq 4 \|\eta\|^{2}_{\infty} \max(\|\partial_x\eta\|^{2}_{\infty},\|\partial_y \eta\|^{2}_{\infty} ) \int_{0}^{t} \mathbb{E}^{\mathbb{Q}}\left[ \|(x_u, y_u)\|^{2}\right] du  \nonumber \\
&= M^{y}_k \int_{0}^{t} \mathbb{E}^{\mathbb{Q}}\left[\|(x_u, y_u)\|^{2}\right] du.
\end{align}
The next step is to estimate $\mathbb{E}^{\mathbb{Q}}\left[(x_t - \bar{x}_t)^{2}\right]$. From the definition of $x_t$, we have that
\begin{equation*}
x_t = \int_{0}^{t} \exp\left(- \int_{u}^{t} k(w) dw\right) \bar{y}_u du + \int_{0}^{t} \exp\left(- \int_{u}^{t} k(w) dw\right) E_{y}(u) du + \int_{0}^{t} \eta(t,x_u,y_u) dW_u. 
\end{equation*}
To find a expansion of $\eta(t,x_u,y_u)$ centered in $\bar{y}_u$, we will define $z^{\nu}_{u} = (1-\nu)\bar{y}_{u} + \nu y_u$. Then
\begin{equation*}
\eta(t,x_u,y_u) = \eta(t,x_u,\bar{y}_u) + \int_{0}^{1}\partial_y\eta(t,x_u,z^{\nu}_{u}) E_{y}(u) d\nu.
\end{equation*}
We will define 
\begin{equation*}
E_{x}\left(u\right) = \int_{0}^{1}\partial_y\eta(t,x_u,z^{\nu}_{u}) E_{y}(u) d\nu. 
\end{equation*}
Therefore, we have the next representation of $x_t$:
\begin{equation*}
x_t = \bar{x}_t + \int_{0}^{t} \exp\left(- \int_{u}^{t} k(w) dw\right) E_{y}(u) du + \int_{0}^{t} \exp\left(- \int_{u}^{t} k(w) dw\right) E_{x}(u) dW_u.
\end{equation*}
Note that
\begin{equation*}
 \mathbb{E}^{\mathbb{Q}}[(E_x(u))^{2}] \leq \|\partial_y \eta \|^{2}_{\infty} \mathbb{E}^{\mathbb{Q}}[(E_{y}(u))^{2}]. 
\end{equation*}
Using (\ref{upper_bound_E_y}), we get that
\begin{equation}\label{upper_bound_E_x}
 \mathbb{E}^{\mathbb{Q}}[(E_x(u))^{2}] \leq 4 \|\eta\|^{2}_{\infty} \|\partial_y \eta\|^{2}_{\infty} \max(\|\partial_x\eta\|^{2}_{\infty},\|\partial_y \eta\|^{2}_{\infty} ) \int_{0}^{u} \mathbb{E}^{\mathbb{Q}}\left[ \|(x_s, y_s)\|^{2}\right] ds.
\end{equation}
By using \eqref{upper_bound_E_x} and \eqref{upper_bound_E_y}, we have that 
\begin{align*}
\mathbb{E}^{\mathbb{Q}}[(x_t - \bar{x}_t)^{2}] & \leq 2 \mathbb{E}^{\mathbb{Q}}\left[ \int_{0}^{t} \exp\left(- \int_{u}^{t} k(w) dw\right) E_{y}(u) du    \right]^2\\
&+ \mathbb{E}^{\mathbb{Q}}\left[ \int_{0}^{t} \exp\left(- \int_{u}^{t} k(w) dw\right) E_{x}(u) dW_u\right]^{2} \\
&\leq 2 \mathbb{E}^{\mathbb{Q}}\left[\int_{0}^{t} \exp\left(- 2\int_{u}^{t} k(w) dw\right) \mathbb{E}^{\mathbb{Q}}(E^{2}_{y}(u)) du\right]\\
&+ \int_{0}^{t} \exp\left(- 2\int_{u}^{t} k(w) dw\right) \mathbb{E}^{\mathbb{Q}}[E^{2}_{x}(u)] du. 
\end{align*}
Join each part and for an adequate constant $M^{x}_k$, we have the following inequality 
\begin{equation*}
\mathbb{E}^{\mathbb{Q}}[(x_t - \bar{x}_t)^{2}]\leq M^{x}_k \int_{0}^{t} \exp\left(- 2\int_{u}^{t} k(w) dw\right) \mathbb{E}^{\mathbb{Q}}\left[ \|(x_u, y_u)\|^{2}\right] du.
\end{equation*}
Finally, we only have to estimate $\mathbb{E}^{\mathbb{Q}}\left[ \|(x_s, y_s)\|^{2}\right]$. From the definition of $y_t$ and given that $\|\eta\|_{\infty} < \infty$, we have that
\begin{equation*}
\mathbb{E}^{\mathbb{Q}}[y^{2}_u] \leq \|\eta\|^{2}_{\infty} \int_{0}^{u} \exp\left(- 2\int_{s}^{u} k(w) dw\right) ds.
\end{equation*}
We must note that at the short term and under the assumptions (\ref{boundedness_reversion}), we get that
\begin{equation*}
\lim_{t \to 0} \mathbb{E}^{\mathbb{Q}}[y^{2}_u] = \|\eta\|^{2}_{\infty} t
\end{equation*}
and 
\begin{equation*}
\lim_{t \to \infty} \mathbb{E}^{\mathbb{Q}}[y^{2}_u] < \infty.
\end{equation*}
Now, we will estimate $\mathbb{E}^{\mathbb{Q}}[x^{2}_u]$. From the definition of $x_t$ and using $(a+b)^2 \leq 2(a^{2}+ b^{2})$, we have that
\begin{align*}
\mathbb{E}^{\mathbb{Q}}[x^{2}_u] &\leq 2 \mathbb{E}^{\mathbb{Q}}\left[\int_{0}^{u} \exp\left(- \int_{s}^{u} k(w) dw\right) y_s  ds \right]^{2} +  2 \mathbb{E}^{\mathbb{Q}}\left[\int_{0}^{u} \exp\left(- \int_{s}^{u} k(w) dw\right) \eta(s,x_s,y_s)  dW_s \right]^{2} \nonumber \\
&\leq 2 \int_{0}^{u} \exp\left(- 2\int_{s}^{u} k(w) dw\right) \mathbb{E}^{\mathbb{Q}}\left[y_s\right]^{2} ds + 2 \int_{0}^{u} \exp\left(- 2 \int_{s}^{u} k(w) dw\right) \mathbb{E}^{\mathbb{Q}}\left[\eta^{2}(s,x_s,y_s)\right] ds \\
&\leq 2 \|\eta\|_{\infty} \int_{0}^{u} \exp\left(- 2\int_{s}^{u} k(w) dw\right) \int_{0}^{s}  \exp\left(- 2\int_{s_1}^{s} k(w) dw\right) ds_1 ds\\
&+  2 \|\eta\|_{\infty} \int_{0}^{u} \exp\left(- 2\int_{s}^{u} k(w) dw\right) ds.
\end{align*}
So as before at the short term
$$
\lim_{t \to 0} \mathbb{E}^{\mathbb{Q}}[x^{2}_u] = 2\|\eta\|^{2}_{\infty} t
$$
and when $u \to \infty$, we have that $\mathbb{E}^{\mathbb{Q}}[x^{2}_u]$ remains bounded. Then, if we use the above inequalities 
\begin{align*}
\mathbb{E}^{\mathbb{Q}}[(y_t - \bar{y}_t)^{2}] &\leq M^{y}_k \int_{0}^{t}  \exp\left(- 2\int_{u}^{t} k(w) dw\right)  \int_{0}^{u} \exp\left(- 2\int_{s}^{u} k(w) dw\right) ds du
\end{align*}
and
\begin{align*}
\mathbb{E}^{\mathbb{Q}}[(x_t - \bar{x}_t)^{2}] &\leq 2 M^{x}_k \|\eta\|_{\infty} \int_{0}^{t} \exp\left(- 2\int_{u}^{t} k(w) dw\right)\int_{0}^{u} \exp\left(- 2\int_{s}^{u} k(w) dw\right) \cdot \\
& \qquad \hspace{1.5cm} \cdot \int_{0}^{s}  \exp\left(- 2\int_{s_1}^{s} k(w) dw\right) ds_1 ds du \\
&+ (2 M^{x}_K \|\eta\|_{\infty} + 1) \int_{0}^{t}  \exp\left(- 2\int_{u}^{t} k(w) dw\right)  \int_{0}^{u} \exp\left(- 2\int_{s}^{u} k(w) dw\right) ds du.
\end{align*}
Therefore, we have that
\begin{align*}
\mathbb{E}^{\mathbb{Q}}[(x_t - \bar{x}_t)^{2}] &\leq \frac{2}{3} M^{x}_k \|\eta\|_{\infty} \left( \int_{0}^{t} \exp\left(- 2\int_{u}^{t} k(w)dw \right) du  \right)^{3}  \\
&+ \frac{(2 M^{x}_K \|\eta\|_{\infty} + 1)}{2} \left( \int_{0}^{t}  \exp\left(- 2\int_{u}^{t} k(w)dw \right) du  \right)^{2}.
\end{align*}
\begin{lemma}\label{f_x_y}[Estimation $f(x_t,y_t)$] Given $f$ continuous and derivable, with $\|\partial_x f\|_{\infty} <\infty$ and $\|\partial_y f\|_{\infty}<\infty$. Then
\begin{equation*}
f(x_t,y_t)=f(\bar{x}_t,\bar{y}_t) + E(t)
\end{equation*}
where $E(t)=\int_{0}^{1} \nabla_{x,y}\ (u, x_{\epsilon,t}, y_{\epsilon,t}) \cdot (x_t - \bar{x}_t,y_t - \bar{y}_t) d\epsilon$ and
\begin{align*}
x_{\epsilon,t} &= \epsilon x_t + (1-\epsilon) \bar{x}_t, \\
y_{\epsilon,t} &= \epsilon y_t + (1-\epsilon) \bar{y}_t.
\end{align*}
In addition, we have that 
\begin{equation*}
\mathbb{E}[E^{2}(t)] \leq \max(\|\partial_x f\|_{\infty}, \|\partial_x f\|_{\infty}) \mathbb{E}^{\mathbb{Q}}[\|(x_t - \bar{x}_t,y_t - \bar{y}_t)\|^{2}].
\end{equation*}
\end{lemma}
\begin{proof}
From the definition of $x_{\epsilon,t}$ and $y_{\epsilon,t}$ we have by the fundamental theorem calculus that

\begin{equation*}
f(x_t,y_t) - f(\bar{x}_t,\bar{y}_t) = \int_{0}^{1} \nabla_{x,y}f (u, x_{\epsilon,t}, y_{\epsilon,t})\cdot (x_t - \bar{x}_t,y_t - \bar{y}_t) d\epsilon
\end{equation*}
Now, if we use Cauchy-Schwarz and boundness of partial derivatives of $f$ we have that
\begin{equation*}
\mathbb{E}^{\mathbb{Q}}[(f(x_t,y_t) - f(\bar{x}_t,\bar{y}_t))^{2}] \leq \max(\|\partial_x f\|_{\infty}, \|\partial_x f\|_{\infty}) \mathbb{E}^{\mathbb{Q}}[\|(x_t - \bar{x}_t,y_t - \bar{y}_t)\|^{2}].
\end{equation*}

\end{proof}

\begin{lemma}\label{DsX}[Approximation $D_s \bar{x}_{t_a}$]
Given the Cheyette model in \eqref{short_rate_cheyette}, the hypotheses \ref{boundedness_volatility} and \ref{boundedness_reversion}, and considering the approximations in \eqref{approximation_y_t} and \eqref{approximation_x_t_a}. Then, 
\begin{align}\label{approximation_D_s_x_t}
D_s \bar{x}_{t_a} &= \beta(s,t_a,\bar{x}_s,\bar{y}_s) \bar{M}(s,t_a,t_p) \quad \text{under the measure $\mathbb{Q}^{t_p}$} \nonumber \\
D_s \bar{x}_{t_a} &= \beta(s,t_a,\bar{x}_s,\bar{y}_s) \bar{M}(s,t_a)  \quad \text{under the measure $\mathbb{Q}$}
\end{align}
where
\begin{align*}
\bar{M}(s,t_a,t_p) &= \exp\left(-\int_{s}^{t_a} \left( \frac{\left(\partial_x \beta(u,t_a,\bar{x}_0,\bar{y}_{t_a})\right)^{2}}{2} - \exp\left(-\int_{u}^{t_a}k_u du\right) \partial_x (\eta(u, \bar{x}_u, \bar{y}_{u}) \bar{\nu}(u,t_p))\right) du \right) \\ 
&\cdot\exp\left(\int_{s}^{t_0} \partial_x \beta(u,t_a,\bar{x}_0,\bar{y}_{t_a}) dW^{\mathbb{Q}}_u \right)
\end{align*}
and
\begin{equation*}
\bar{M}(s,t_a) = \exp\left(-\int_{s}^{t_a} \left( \frac{\left(\partial_x \beta(u,t_a,\bar{x}_0,\bar{y}_{t_a})\right)^{2}}{2} - \right) du \right) \exp\left(\int_{s}^{t_0} \partial_x \beta(u,t_a,\bar{x}_0,\bar{y}_{t_a}) dW^{\mathbb{Q}}_u \right)
\end{equation*} 
with $\beta(u,t_a,x,y) = \exp\left(-\int_{u}^{t_a}k(w)dw\right) \partial_x \eta(u,x,y)$.\\

We have also that
\begin{align}\label{approsimation_E_s_x_t}
\mathbb{E}_s^{t_p}\left[D_s \bar{x}_{t_a}\right] &\approx \beta(s,t_a,\bar{x}_s,\bar{y}_s) \overline{DM}(s,t_a), \nonumber \\
\mathbb{E}_s^{\mathbb{Q}}\left[D_s \bar{x}_{t_a}\right]&\approx \beta(s,t_a,\bar{x}_s,\bar{y}_s)
\end{align}
where 
\begin{align*}
\overline{DM}(s,t_a) &= \mathbb{E}^{t_p}_s\biggl[\exp\biggr(-\int_{s}^{t_a} \exp\Bigl(-\int_{u}^{t_a}k_u du\Bigr) \Bigl(\partial_x \eta(u, \bar{x}_0(t_a), \bar{y}_{u}) \bar{\nu}(u,t_p, \bar{x}_0(t_a)) + \eta(u,\bar{x}_0(t_a), \bar{y}_{u}) \\
\qquad &\hspace{3cm} \cdot \partial_x \bar{\nu} (t,t_p, \bar{x}_0(t_a))\Bigr) du\biggr)\biggr]
\end{align*}
with  
\begin{align*}
\bar{\nu}(t,t_p) &= \int_{t}^{t_p} \eta(s,\bar{x}_s,\bar{y}_s) ds, \\
\bar{\nu}(t,t_p, x_0) &= \int_{t}^{t_p} \eta(s,x,\bar{y}_s) ds,\\
\intertext{and}
\partial_x \bar{\nu}(t,t_p,x_0) &= \int_{t}^{t_p} \eta(s,x_0,\bar{y}_s) ds.
\end{align*}
\end{lemma}
\begin{proof}
From 
\begin{equation*}
dW^{\mathbb{Q}^{t_p}} = dW^{\mathbb{Q}} + \bar{\nu}(t,t_p) dt
\end{equation*}
we have that under the measure $\mathbb{Q}^{t_p}$ 
\begin{align*}
\bar{x}_{t_a} = \bar{x}_0(t_a) &+ \int_{0}^{t_a} \exp\left(-\int_{s}^{t_a}k(w)dw\right) \bar{y}_s ds - \int_{0}^{t_a} \exp\left(-\int_{s}^{t_a}k(w)dw\right) \bar{\nu}(s, t_p) \eta(s,\bar{x}_s,\bar{y}_s) ds   \\
&+ \int_{0}^{t_a}  \exp\left(-\int_{s}^{t_a}k(w)dw \right)\eta(s,\bar{x}_s,\bar{y}_s) dW_s^{\mathbb{Q}^{t_p}}. 
\end{align*}
Now, if we apply $D_s$ in the above equality and we use the last approximation we have that
\begin{align*}
\mathbb{E}^{t_p}_s[D_s \bar{x}_{t_a}] &\approx \eta(s,\bar{x}_s,\bar{y}_s) \exp\left(-\int_{s}^{t_a}k(w)dw \right)\\ &\cdot \exp\biggr(-\int_{s}^{t_a} \exp\Bigl(-\int_{u}^{t_a}k(w)dw\Bigr) \Bigl(\partial_x \eta(u, \bar{x}_0(t_a), \bar{y}_{u}) \bar{\nu}(u,t_p, \bar{x}_0(t_a)) + \eta(u,\bar{x}_0(t_a), \bar{y}_{u}) \\
\qquad &\hspace{3cm} \cdot \partial_x \bar{\nu} (t,t_p, \bar{x}_0(t_a))\Bigr) du\biggr).
\end{align*}
\end{proof}

\subsection{Approximation of $\mathbb{E}_s^{\mathbb{Q}}\left[D_s \bar{x}_{t_a}\right]$}
As in the previous appendix, we have
\begin{align*}
\bar{x}_{t_a} &= \bar{x}_0(t_a) + \int_{0}^{t_a} \exp\left(-\int_{s}^{t_a}k(w)dw\right) \bar{y}_u du + \int_{0}^{t_a} \exp\left(-\int_{s}^{t_a}k(w)dw \right) \eta(u,\bar{x}_u,\bar{y}_u) dW_u^{\mathbb{Q}}
\end{align*}
and therefore (see (\ref{approsimation_E_s_x_t}))
\begin{equation}
D_s \bar{x}_{t_a} =  \exp\left(-\int_{s}^{t_a}k(w)dw \right) \eta(s,\bar{x}_s,\bar{y}_{s})\bar{M}(s,t_a).
\end{equation}
Now, if we take $\mathbb{E}_s^{\mathbb{Q}}\left[\cdot\right]$, we get
\begin{equation}\label{approximation_spot_E_s_x_t}
\mathbb{E}^{\mathbb{Q}}_s\left(D_s \bar{x}_{t_a} \right)=\exp\left(-\int_{s}^{t_a}k(w)dw \right) \eta(s,\bar{x}_s,\bar{y}_{s}).
\end{equation}
Then, we obtain that
\begin{equation}\label{approximation_spot_E_s_x_t}
\mathbb{E}^{\mathbb{Q}}_s\left(D_s \bar{x}_{t_a} \right) \approx \exp\left(-\int_{s}^{t_a}k(w)dw \right) \eta(s,\bar{x}_0(t_a),\bar{y}_{s}).
\end{equation}

\subsection{Approximation of $\mathbb{E}^{01}\left[\bar{x}_{t_a}\right]$}
It is easy to show that the bond dynamics under the HJM assumption is
\begin{equation}
\frac{dP(t,T)}{P(t,T)} = r_t dt - \nu(t,T)dW^{\mathbb{Q}}_t  
\end{equation}
where we must remember that  $\nu(t,T)=\int_{t}^{T}\sigma(t,s) ds$. Therefore, if we apply the Itô formula we have that 
\begin{equation}\label{annuity_spot_dynamic}
\frac{d01(t,t_a,t_b)}{01(t,t_a,t_b)} = r_t dt - \sigma_{01}(t,t_a,t_b)dW^{\mathbb{Q}}_t
\end{equation}
with 
\begin{equation} \label{annuity_vol}
\sigma_{01}(t,t_a,t_b) = \frac{\sum_{i=a+1}^{b} \delta_{i-1,i} P(t,t_i) \nu(t,t_i)}{01(t,t_a,t_b)}.
\end{equation}
We define $w_i(t)=\frac{\delta_{i-1,i} P(t,t_i)}{01(t,t_a,t_b)}$, then
\begin{equation*}
\sigma_{01}(t,t_a,t_b) = \sum_{i=a+1}^{b}  w_i(t) \nu(t,t_i).
\end{equation*}
From (\ref{annuity_vol}) and since $\frac{P_{ois}(t,T)}{01(t)}$ is a martingale, we have that
\begin{equation*}
dW^{01}_t = dW^{\mathbb{Q}}_t - \sigma_{01}(t,t_a,t_b)dt. 
\end{equation*}
Then, if we freeze the weights $w_i(t)$, we get the next approximation of (\ref{annuity_vol})
\begin{equation*} \label{approximation_o1_vol}
\bar{\sigma}_{01}(t,t_a,t_b) \approx \sum_{i=a+1}^{b} w_i(0) \nu(t,t_i).
\end{equation*}
Using (\ref{short_rate_cheyette}) and the above approximation, we obtain that
\begin{align}
\bar{x}_{t_a} &\approx \bar{x}_0(t_a)  + \int_{0}^{t_a} \exp\left(-\int_{s}^{t}k(w)dw\right) \bar{y}_u du + \int_{0}^{t_a}  \exp\left(-\int_{s}^{t_a}k(w)dw \right) \eta(u,\bar{x}(u),\bar{y}_u) dW_u^{\mathbb{Q}} \nonumber \\
&=  \bar{x}_0(t)  + \int_{0}^{t_a} \exp\left(-\int_{s}^{t_a}k(w)dw\right) \bar{y}_u du + \int_{0}^{t_a} \exp\left(-\int_{s}^{t_a}k(w)dw \right) \eta(u,\bar{x}_u,\bar{y}_u) \bar{\sigma}_{01}(u, t_a, t_b) du  \nonumber \\ 
&\hspace{1.2cm} + \int_{0}^{t_a} \exp\left(-\int_{s}^{t_a}k(w)dw\right) \eta(u,\bar{x}_u,\bar{y}_u) dW_u^{01}.
\end{align}
Therefore, we have that
\begin{align}
\mathbb{E}^{01}\left[\bar{x}_t\right] &\approx \bar{x}_0(t_a)  + \int_{0}^{t_a} \exp\left(-\int_{s}^{t_a}k(w)dw\right) \bar{y}_u du \nonumber \\ 
&\hspace{1.3cm} + \int_{0}^{t_a} \exp\left(-\int_{s}^{t_a}k(w)dw\right) \eta(u,\bar{x}_u,\bar{y}_u) \bar{\sigma}_{01}(u, t_a, t_b,\bar{x}_0(t_a)) du
\end{align}
with
$$
\bar{\sigma}_{01}(u, t_a, t_b,\bar{x}_0(t_a)) = \sum_{i=a+1}^{b} w_i(0) \nu(t,t_i,\bar{x}_0(t_a))
$$
and
$$
\nu(t,t_i,\bar{x}_0(t_a)) = \int_{t}^{t_i} \sigma(t,u,\bar{x}_0(t_a),\bar{y}_u) du.
$$ 

\subsection{Approximation of $\mathbb{E}^{01}\left[D_s \bar{x}_{t_a} \right]$}\label{approximation_under_annuity_measeure_d_s}
Let us to remember that $\nu(t,t_i)=h(t,\bar{x}_t,\bar{y}_t) \frac{G(t,t_i)}{\beta_{t,k}(t)}$ with  $\beta_{t,k}(t) = \exp\left(\int_{0}^{t}k(w)dw\right)$.
Therefore 
\begin{equation*}
D_s \nu(t,T_i) = \partial_x \nu(t,T_i, \bar{x}_t, \bar{y}_t) D_s \bar{x}_t  \frac{G(t,t_i)}{\beta_{t,k}}.
\end{equation*}
Then, we have that
\begin{equation}
D_s \sigma_{0,1}(t,t_a,t_b) =  D_s \bar{x}_t \mu(t,\bar{x}_t, \bar{y}_t, t_a,t_b)
\end{equation}
where 
$$
\mu(t,\bar{x}_t, \bar{y}_t, t_a,t_b) =  \sum_{i=a+1}^{b} w_i(t) \partial_x \nu(t,t_i, \bar{x}_t, \bar{y}_t) \frac{G(t,t_i)}{\beta_{t,k}},
$$
and $\frac{\delta_{i-1,i} P(t,t_i)}{01(t,t_a,t_b)}$. From (\ref{approximation_x_t_a})  and the Girsanov's theorem, we get that
\begin{align*}
\bar{x}_{t_a} =  \bar{x}_0(t_a)  &+ \int_{0}^{t_a} \exp\left(-\int_{s}^{t_a}k(w)dw\right) \bar{y}_u du \nonumber \\\ 
&+ \int_{0}^{t_a} \exp\left(-\int_{s}^{t_a}k(w)dw\right) \sigma_{01}(u,t_a,t_b) \mu(u,\bar{x}_u, \bar{y}_u, t_a,t_b) du \nonumber \\\
&+  \int_{0}^{t_a}  \exp\left(-\int_{s}^{t_a}k(w)dw \right) \eta(u,\bar{x}_u,\bar{y}_u) dW_u^{\mathbb{Q}^{01}}.
\end{align*}
Then, taking $D_s$, we have that
\begin{equation*}
D_s \bar{x}_{t_a} = \exp\left(-\int_{s}^{t_a}k(w)dw \right) \eta(s,\bar{x}_{t_a},\bar{y}_{t_a})\bar{M}^{01}(s,t_a)
\end{equation*}
where
\begin{align*}
\bar{M}^{01}(s,t_a) = &\exp\left(-\int_{s}^{t_a} \left(\frac{\left(\partial_x \beta(u,t_a,\bar{x}_u,\bar{y}_{u})\right)^{2}}{2} + \partial_x (\beta(u,t_a,\bar{x}_u,\bar{y}_{u}) \mu(u,\bar{x}_u,\bar{y}_u,t_a,t_b)) du \right)\right) \\
&\exp\left(\int_{s}^{t_a} \partial_x \beta(u,t_a,\bar{x}_0,\bar{y}_{t_a}) dW^{\mathbb{Q}^{01}}_u \right)
\end{align*}
and $\beta(u,t_a,x,y) = \exp\left(-\int_{u}^{t_a}k(w)dw\right)\eta(u,x,y)$.\\
Then, 
\begin{align}\label{approximation_E_01_Ds_x_t}
\mathbb{E}^{01}_s\left[D_s \bar{x}_{t_a}\right] &= \beta(s,t_a,\bar{x}_s,\bar{y}_s)  \mathbb{E}_s^{01}\left(\exp\left(-\int_{s}^{t_a}\partial_x (\beta(u,t_a,\bar{x}_u,\bar{y}_u) \mu(u,\bar{x}_u, \bar{y}_u,t_a,t_b)) du \right)\right) \nonumber \\
&\approx \beta(s,t_a,\bar{x}_s,\bar{y}_s) \exp\left(-\int_{s}^{t_a}\partial_x (\beta(u,t_a,\bar{x}_u,\bar{y}_u) \mu(u,\bar{x}_u, \bar{y}_u,t_a,t_b))|_{\bar{x}_u=\bar{x}_{0}(t_a)}  du \right)
\end{align}

\section{Proofs}
\subsection{Proof Theorem \ref{Th_CA_futures}}\label{Proof_CA_futures}
%Observe that $L_{E}(t, t_1, t_2)$ is a martingale under the forward measure $\mathbb{Q}^{t_2}$. Let us define the future rate as:    
%\begin{equation}\label{future}
%\hat{L}_{E}(t,t_0, t_1, t_2) = \mathbb{E}_t^{\mathbb{Q}}\left[L_{E}(t_0, t_1, t_2) \right]
%\end{equation}
%where $\mathbb{Q}$ is the measure associated to the numeraire $B_t=\exp\left(\int_{0}^{t} r_{ois, s} ds \right)$ with $ r_{ois, t}$ the risk free short rate. Using
%\eqref{forward_rate} and \eqref{future}, then the convexity adjustment definition is
%\begin{equation*}
%CA(t, t_0, t_1, t_2) = \hat{L}_{E}(t,t_0, t_1, t_2) - \mathbb{E}_t^{\mathbb{Q}^{t_2}}\left[L_{E}(t_0, t_1, t_2) \right].
%\end{equation*}
Note that
\begin{equation*}
CA(t, t_0, t_1, t_2) = \hat{L}_{E}(t,t_0, t_1, t_2) - \mathbb{E}_t^{\mathbb{Q}^{t_2}}\left[L_{E}(t_0, t_1, t_2) \right].
\end{equation*}
From (\ref{ois_forward_rate_curve}) and since $f_{ois}(t,T)$ is a $\mathbb{Q}^{T}$ martingale, we have that
\begin{equation}\label{girsanov_spot_forward}
dW^{\mathbb{Q}^{t_2}} = dW^{\mathbb{Q}} + \nu(t,t_2) dt. 
\end{equation}
Applying (\ref{general_convexity}) with $f(x_t)=L_{E}(t,t_0, t_1, t_2)$, $\mathbb{Q}_1=\mathbb{Q}$, $\mathbb{Q}_1=\mathbb{Q}^{t_2}$ and $\lambda_t = \nu(t,t_2)$,  we get that
\begin{equation}\label{ca_general_future}
CA(t, t_0, t_1, t_2) = \mathbb{E}^{\mathbb{Q}^{t2}}\left[\int_{0}^{t_0} \mathbb{E}^{\mathbb{Q}}_{s}\Bigl[D_s L_{E}(t_0,t_1,t_2) \Bigr] \nu(s,t_2) ds \right]
\end{equation}
where $\nu(t,T)$ has been defined in (\ref{ois_forward_rate_curve}). Calculating the Malliavin derivative of $L_{E}(t_0,t_1,t_2)$ we have that
\begin{equation*}
D_s L_{E}(t_0,t_1,t_2) = \frac{H(t_0,t_1)}{\delta_{t_1,t_2}H(t_0,t_2)} D_s \left(\frac{P_{ois}(t_0,t_1)}{P_{ois}(t_0,t_2)}\right).
\end{equation*}
Now from the zero-coupon representation formula (\ref{bond_ois}), we get that
\begin{equation*}
D_s \left(\frac{P_{ois}(t_0,t_1)}{P_{ois}(t_0,t_2)}\right) = \frac{\partial_{x}P_{ois}(t_0,t_1)P_{ois}(t_0,t_2) - \partial_{x}P_{ois}(t_0,t_2) P_{ois}(t_0,t_1)}{P^{2}_{ois}(t_0,t_2)} D_s x_{t_0}.
\end{equation*}
Therefore
\begin{equation}\label{malliavin_derive_L}
D_s L_{E}(t_0,t_1,t_2) = \frac{H(t_0,t_1)}{\delta_{t_1,t_2}H(t_0,t_2)}\frac{\partial_{x}P_{ois}(t_0,t_1)P_{ois}(t_0,t_2) - \partial_{x}P_{ois}(t_0,t_2) P_{ois}(t_0,t_1) }{P^{2}_{ois}(t_0,t_2)} D_s x_{t_0}.
\end{equation}
If we use (\ref{approsimation_E_s_x_t}) with $T_a=t_0$ and $\beta(t,t_0,x,y) = \exp\left(-\int_{s}^{T_a}k(w)dw \right) \eta(u,x,y)$, we have that
\begin{align*}
D_s L_{E}(t_0,t_1,t_2) &\approx \frac{H(t_0,t_1)}{\delta_{t_1,t_2}H(t_0,t_2)}\frac{\partial_{x}P_{ois}(t_0,t_1)P_{ois}(t_0,t_2) - \partial_{x}P_{ois}(t_0,t_2) P_{ois}(t_0,t_1)}{P^{2}_{ois}(t_0,t_2)} \beta(t,t_0,\bar{x}_s,\bar{y}_s)\bar{M}(s,t_0) \nonumber \\
\approx& \frac{P_{E}(0,t_1)}{\delta_{t_1,t_2} P_{E}(0,t_2)} \left(G(t_0,t_2) - G(t_0,t_1) \right)\beta(t,t_0,\bar{x}_s,\bar{y}_s)\bar{M}(s,t_0).
\end{align*}
Therefore
\begin{equation}\label{approximation_clarkocone}
\mathbb{E}_s\left[ D_s L_{E}(t_0,t_1,t_2) \right] = \frac{P_{E}(0,t_1)}{\delta_{t_1,t_2} P_{E}(0,t_2)} \Bigl(G(t_0,t_2)  - G(t_0,t_1)\Bigr)\beta(s,t_0,x_0,\bar{y}_s).
\end{equation}
Then from (\ref{ca_general_future}) and (\ref{approximation_clarkocone}) we find the approximation for the convexity adjustment for futures.

\subsection{Proof Theorem \ref{Th_CA_OIS}}\label{Proof_CA_OIS_futures}
To compute $\mathbb{E}^{\mathbb{Q}}\left[R(t_0,t_1)\right]$ we will use (\ref{general_convexity}), with $\mathbb{Q}_2 = \mathbb{Q}$, $\mathbb{Q}_1 = \mathbb{Q}^{t_1}$ and $x_{t_1}=R(t_0,t_1)$. If we apply $D_s$ on $I(t_0,t_1)$ we obtain that
\begin{equation*}
D_s I(t_0,t_1) = \int_{\max(s, t_{0})}^{t_1}  D_s x_u du.
\end{equation*}
Now, if $t < t_0$, then from (\ref{clark-okone}) and (\ref{approximation_spot_E_s_x_t}), we have that 
\begin{align}\label{apprx_I_t0_t1}
I(t_0,t_1) &= \mathbb{E}_t^{\mathbb{Q}}\left[I(t_0,t_1)\right] + \int_{t}^{t_1}\int_{\max(s, t_{0})}^{t_1}  \mathbb{E}_s^{\mathbb{Q}}\left[\beta(s,u,x_s,y_s) \bar{M}(s,u) \right] du dW_s^{\mathbb{Q}} \nonumber \\
&\approx \mathbb{E}_t^{\mathbb{Q}}\left[I(t_0,t_1)\right] + \int_{t}^{t_1}\int_{\max(s, t_{0})}^{t_1} \beta(s,u,\bar{x}_0(t_1),\bar{y}_s) du dW_s^{\mathbb{Q}}\nonumber \\
&= \mathbb{E}_t^{\mathbb{Q}}\left[I(t_0,t_1)\right] + \int_{t}^{t_1} g(s)h(s,\bar{x}_0(t_1),\bar{y}_s)\int_{\max(s, t_{0})}^{t_1} \exp\left( -\int_{s}^{u} k(w)dw \right) du dW_s^{\mathbb{Q}}.
\end{align}
Then, using the previous approximation, we get that
\begin{align*}
1 + \delta_{t_0,t_1}\mathbb{E}_t^{\mathbb{Q}}\left[R(t_0,t_1)\right] &= \mathbb{E}_t^{\mathbb{Q}}\left[ \exp(I(t_0,t_1)) \right]\\
&\approx \exp\left(\mathbb{E}_t^{\mathbb{Q}}\left[I(t_0,t_1)\right]\right)  \mathbb{E}^{\mathbb{Q}}\left[\exp\left(\int_{t}^{t_1} \Gamma(s,t_0,t_1)dW_s^{\mathbb{Q}}\right)\right]
\end{align*}
where 
\begin{align*}
\Gamma(s,t_0,t_1)= g(s)h(s,x_0,y_0)\int_{\max(s, t_{0})}^{t_1} \exp\left( -\int_{s}^{u} k(w)dw \right)du.
\end{align*}
Therefore, we have that
\begin{equation}\label{representation_i_0_t}
1 + \delta_{t_0,t_1}\mathbb{E}^{\mathbb{Q}}\left[R(t_0,t_1)\right] \approx \exp\left(\mathbb{E}^{\mathbb{Q}}\left[I(t_0,t_1)\right]\right)\exp\left(-\frac{1}{2}\int_{t}^{t_1}\Gamma^{2}(s,t_0,t_1) ds\right).
\end{equation}
Then, we obtain \eqref{convexity_ois_future}.\\

In order to get an approximation of $\mathbb{E}^{\mathbb{Q}}\left[R_{avg}(t_0,t_1)\right]$ with base $\mathbb{E}^{\mathbb{Q}}\left[R(t_0,t_1)\right]$, we must note that 
\begin{equation*}
\mathbb{E}_t^{\mathbb{Q}}\left[R_{avg}(t_0,t_1)\right] = \frac{\mathbb{E}_t^{\mathbb{Q}}\biggl[ \log\Bigl(1+ \delta_{t_0,t_1} R(t_0,t_1)\Bigr)\biggr]}{ \delta_{t_0,t_1}}.
\end{equation*}
Then from (\ref{representation_i_0_t}), we get \eqref{convexity_avg_ois_future}.

\subsection{Proof Theorem \ref{Th_CA_FRAsinArrears}}\label{Proof_CA_FRAsinArrears}
$L_{E}(t,t_1,t_2)$ is martingale under the measure $\mathbb{Q}^{t_2}$, therefore the expected value of \eqref{FRAinArrear} is taken with respect to the wrong martingale. To calculate the convexity adjustment, we use the Clark-Ocone formula to get a representation for $L_{E}(t_1,t_1,t_2)$, i.e
\begin{equation}\label{general_convexity_fras}
L_{E}(t_1,t_1,t_2) = \mathbb{E}^{t_2}\Bigl[L_{E}(t_1,t_1,t_2) \Bigr] + \int_{0}^{t_1} \mathbb{E}^{t_2}\Bigl[D_s L_{E}(t_1,t_1,t_2) \Bigr] dW^{\mathbb{Q}^{t_2}}_{s}.
\end{equation}
Under the HJM dynamics, we have the relation
\begin{equation*}
dW^{\mathbb{Q}^{t_2}}_s = dW^{\mathbb{Q}^{t_1}}_s + \Bigl(\nu(s,t_2) - \nu(s,t_1)\Bigr) ds. 
\end{equation*}
Taking $\mathbb{E}^{^{t_1}}(\cdot)$, we get that
\begin{align*}
\mathbb{E}^{t_1}\left[L_{E}(t_1,t_1,t_2) \right] &= L_{E}(0,t_1,t_2) + \mathbb{E}^{t_1}\left[ \int_{0}^{t_1} \mathbb{E}_s^{t_2}\Bigl[D_s L_{E}(t_1,t_1,t_2)\Bigr] dW^{t_2}_{s}\right]\\
&= L_{E}(0,t_1,t_2) + \mathbb{E}^{t_1}\left[ \int_{0}^{t_1} \mathbb{E}_s^{t_2}\Bigl[D_s L_{E}(t_1,t_1,t_2) \Bigr] \Bigl(\nu(s,t_2)-\nu(s,t_1)\Bigr) ds \right].
\end{align*}
Now from (\ref{approximation_D_s_x_t}) we have that 
\begin{eqnarray*}
D_s L(t_1,t_1,t_2) &=& \frac{G(t_1,t_2)}{\delta_{t_1,t_2}P_{E}(t_1,t_2)} D_s x_{t_1}\\
&\approx& \frac{G(t_1,t_2)}{\delta_{t_1,t_2}P_{E}(0,t_1,t_2)} \beta(s,t_1, \bar{x}_0(t_1), \bar{y}_s)\overline{DM}(s,t_1).
\end{eqnarray*}
Using the above approximation, we get
\begin{equation*}
CA(t_0,t_1,t_2) \approx \frac{G(t_1,t_2)}{\delta_{t_1,t_2}P_{E}(0,t_1,t_2)} \int_{0}^{t_1} \beta(s,t_1, \bar{x}_0(t_1), \bar{y}_s) \overline{DM}(s,t_1) \Bigl(\bar{\nu}(s,t_2,\bar{x}_{0}(t_1))-\bar{\nu}(s,t_1,\bar{x}_{0}(t_1))\Bigr) ds
\end{equation*}
with 
\begin{equation*}
\bar{\nu}(t,t_p, \bar{x}_0(t_a))= \int_{t}^{t_p} \eta(s,\bar{x}_0(t_a),\bar{y}_s) ds.   
\end{equation*}

\subsection{Proof Theorem \ref{Th_CA_CMS}}\label{Proof_CA_CMS}
Assume we have a cash flow in $t_a < t_p < t_b$ with value $S_{a,b}(t_a)$. Recall that $S_{a,b}(t_a)$ is a martingale under the measure $\mathbb{Q}^{01}$, but not under the measure $\mathbb{Q}^{t_p}$. Therefore, we take into consideration the effect to compute the expected value of $S_{a,b}(t_a)$ in a measure that is not its natural measure. Then, the convexity adjustment for a CMS is
\begin{equation}
CA_{CMS}(t_p) = \mathbb{E}^{t_p}\left[S_{a,b}(t_a)\right] - S_{a,b}(0).
\end{equation} 
After some changes of measure, we can see that
\begin{align}
\mathbb{E}^{t_p}\left(S_{a,b}(t_a)\right) &= \frac{1}{M(0,t_p)} \mathbb{E}^{01}\left[S_{a,b}(t_a) M(t_a,t_p)\right] \nonumber \\
&= \frac{1}{M(0,t_p)} \mathbb{E}^{01}\biggl[S_{a,b}(t_a) \mathbb{E}^{01}\Bigl[ M(t_a,t_p)|S_{a,b}(t_a)\Bigr] \biggr] \nonumber
\end{align}
with 
\begin{align*}
M(t,t_p)= \frac{P_{ois}(t,t_p)}{01(t,t_a,t_p)}.
\end{align*}
Then, we can approximate $\mathbb{E}^{t_p}\left(S_{a,b}(t_a)\right)$ by
\begin{align}\label{expected_t_p_cms}
\mathbb{E}^{t_p}\left(S_{a,b}(t_a)\right) &\approx  \frac{1}{M(0,t_p)} \mathbb{E}^{01}\biggl[S^{ois}_{a,b}(t_a) \mathbb{E}^{01}\Bigl[ M(t_a,t_p)|S_{a,b}(t_a)\Bigr] \biggr] \\
\qquad\qquad &+ \frac{\sum_{i=1}^{n_E}\delta_{t^{E}_{i-1}, t^{E}_i} \alpha(0,t^{E}_{i-1}, t^{E}_{i}) P_{ois}(0,t^{E}_{i})} {01(0,t_a,t_b)}.
\end{align}

From the previous expression, we must note that under the assumption that there are not stochastic basis, we must compute the convexity adjustment for the OIS swap rate. But a complicated point is to calculate the expected value 
\begin{equation}\label{Exp_M_conditional}
\mathbb{E}^{01}\Bigl[ M(t_a,t_p)|S_{a,b}(t_a)\Bigr].
\end{equation}
To reduce this complexity, it is a common practice to assume that $M(t_a,t_p)$ is a function of the swap rate $S_{a,b}(t_a)$, i.e   $M(t_a,t_p)=f(S_{a,b}(t_a))$. Under this assumption (\ref{Exp_M_conditional}) is trivial to calculate it. The function $f(\cdot)$ is known as the mapping function. There is a vast literature about how to choose it (see \cite{AndreasenPiterbargIII} or \cite{Hagan20}). In the present paper, we will not assume any mapping function. \\ 
We will choose $\bar{x}_0(t_a)$ such that $S^{ois}_{a,b}(t_a,\bar{x}_0(t_a),\bar{y}_{t_a})=S^{ois}_{a,b}(t_a,0,0)$. Now, if we apply the Clark-Ocone formula to $M(t_a,t_p)$ we get that
\begin{equation} \label{clark_ocone_swap_m}
M(t_a,t_p) = M(0,t_p)+ \int_{0}^{t_a} \mathbb{E}_s^{01}\Bigl[D_s M(t_a,t_p)\Bigr] dW^{01}_s.
\end{equation}
Then, if we substitute the last expressions in (\ref{expected_t_p_cms}), we obtain that
\begin{align}\label{cms_first_order_convexity_prev}
\mathbb{E}^{t_p}\left(S_{a,b}(t_a)\right) =& S^{ois}_{a,b}(0) + \frac{1}{M(0,t_p)} \mathbb{E}^{01}\left[ S^{ois}_{a,b}(t_a) \int_{0}^{t_a} \mathbb{E}^{0,1}_s\Bigl[D_s x_{t_a}\partial_x M(t_a,t_p)  \Bigr] dW^{0,1}_s   \right] \nonumber \\
=&  S^{ois}_{a,b}(0) + \frac{1}{M(0,t_p)} \mathbb{E}^{01}\left[\int_{0}^{t_a} D_s x_{t_a} \partial_x S^{ois}_{a,b}(t_a) \mathbb{E}^{0,1}_s\Bigl[D_s x_{t_a}\partial_x M(t_a,t_p)  \Bigr] ds   \right].
\end{align}
Now from (\ref{approximation_E_01_Ds_x_t}) and choosing $\bar{x}_0(t_a)$ such that $S^{ois}_{a,b}(t_a,\bar{x}_0(t_a),\bar{y}_{t_a})=S^{ois}_{a,b}(0)$, we have that
\begin{align}
\mathbb{E}^{t_p}\left[S_{a,b}(t_a)\right] &\approx  S^{ois}_{a,b}(0) + \frac{\partial_x S^{ois}_{a,b}(t_a,\bar{x}_0(t_a), \bar{y}_{t_a}) \partial_x M(t_a,t_p,\bar{x}_0(t_a), \bar{y}_{t_a})} {M(0,t_p)}\nonumber \\
&\cdot \int_{0}^{t_a}  \beta^2(s,t_a,\bar{x}_0(t_a),\bar{y}_s)\exp\left(-\int_{s}^{t_a}\partial_x (\beta(u,t_a,\bar{x}_u,\bar{y}_u) \mu(u,\bar{x}_u, \bar{y}_u,t_a,t_b))|_{\bar{x}_u=\bar{x}_{0}(t_a)}  du \right)ds.
\end{align}

% TODO: Change bibliographystyle according to the journal style - it should know the online entry, see Matsuda04
\bibliography{references/references, references/references-books,references/references-own,references/references-online}
%\bibliography{references-export}

\end{document}